\documentclass[11pt]{article}

\usepackage[top=2cm]{geometry}

\usepackage[utf8]{inputenc}
\usepackage{amsthm}
\usepackage{amsmath}
\usepackage{tikz}
\usepackage{mathtools}
\usepackage{amssymb}
\usepackage{amsfonts}
\usepackage{parskip}
\usepackage{mathrsfs}
\usepackage{bbm}
\usepackage{setspace}
\usepackage{caption}
\usepackage{subcaption}
\usepackage[round]{natbib}
\usepackage{centernot}
\usepackage{float}
\usetikzlibrary{positioning}
\newtheorem{lemma}{Lemma}
\newtheorem*{lemma A}{Lemma A}
\newtheorem{corollary}{Corollary}
\newtheorem{theorem}{Theorem}
\newtheorem{remark}{Remark}
\newtheorem{assumption}{Assumption}

\newtheorem{proposition}{Proposition}
\newtheorem*{proposition*}{Proposition}
\newtheorem{example}{Example}
\newtheorem{fact}{Fact}
\allowdisplaybreaks
\usepackage{verbatim}
\usepackage{hyperref}
\usepackage{soul}

\usetikzlibrary{arrows.meta}
\usetikzlibrary{fit,positioning}

   \usepackage{xcolor}
\hypersetup{colorlinks=true,linkcolor=blue,urlcolor=blue,citecolor=purple} 

\renewcommand{\b}{\mathbf}
\renewcommand{\c}[1]{\mathcal{#1}}
\newcommand{\hb}[1]{\hat{\mathbf{#1}}}

\newcommand{\lmax}{\lambda_1}
\newcommand{\lmin}{\lambda_n}
\newcommand{\pmat}[1]{\begin{pmatrix}#1\end{pmatrix}}
\newcommand{\bmat}[1]{\begin{bmatrix}#1\end{bmatrix}}

\renewenvironment{proof}[1]{{\bf{Proof of #1.} }}{\hfill $\Box$} 
\bibpunct[: ]{(}{)}{;}{a}{}{,}

\begin{document}

\title{\sc  Welfare and Distributional Effects of Joint Intervention in Networks\footnote{We are grateful to the editor Marzena Rostek, an advisory editor, two anonymous referees, Francis Bloch, Yann Bramoulle, Antonio Cabrales, Vasco M. Carvalho, Krishna Dasaratha, Sihua Ding, Matt Elliott, Andrea Galeotti, Sanjeev Goyal,  Cheng-Zhong Qin,  Evan Sadler, Yang Sun,  Alireza Tahbaz-Salehi,   Guofu Tan,  Fernando Vega-Redondo, Wei Zhao and seminar participants for very helpful comments. Junjie Zhou gratefully acknowledges financial support from the NSFC (Grant Nos. 72450001,72521003, and 72342032) and the Tsinghua Strategy for Heightening Arts, Humanities, and Social Sciences: “Plateaus \& Peaks” (No. 2022TSG08102). The usual disclaimer applies.}
}
\author{
Ryan Kor\thanks{Department of Economics, National University of Singapore.  {\tt e0004083@u.nus.edu}}
\and Yi Liu\thanks{Department of Economics, Yale University.  {\tt yi.liu.yl2859@yale.edu}}
\and Yves Zenou\thanks{Department of Economics, Monash University.  {\tt yves.zenou@monash.edu}}
 \and Junjie Zhou\thanks{School of Economics and Management, Tsinghua University.  {\tt zhoujj03001@gmail.com}} }
\date{\today}

\maketitle

\begin{abstract}
We study the optimal joint intervention of a planner who can influence both the standalone marginal utilities of agents in a network and the weights of the links connecting them. The welfare-maximizing intervention displays two key features. First, when the planner’s budget is moderate (yielding interior solutions), the optimal change in link weight between any pair of agents is proportional to the product of their eigen-centralities. Second, when the budget is sufficiently large, the optimal network converges to a simple structure: a complete network under strategic complements, or a complete balanced bipartite network under strategic substitutes. We show that welfare effects are governed by the principal eigenvalue of the network, while distributional outcomes are driven by the dispersion of the corresponding eigen-centralities. Comparing joint interventions to single interventions targeting only standalone marginal utilities, we find that joint interventions consistently generate higher aggregate welfare, but may also increase inequality, revealing a potential trade-off between efficiency and equity.
\medskip

\noindent \textbf{Classification JEL:} D21, D29, D82.

\end{abstract}

\newpage

\section{Introduction}

In many socio-economic settings, individual behaviors are embedded in networks   where one person’s actions directly influence others. For a given network structure, a planner can intervene along two main dimensions: by altering individual incentives  or by reshaping the network itself. In practice, effective policies often combine both. Consider infrastructure networks: to reduce pollution or congestion, a planner may adjust incentives through fuel taxes, parking fees, congestion charges, or public transport subsidies, while simultaneously modifying the network through road expansions, dedicated bus lanes, traffic-light coordination, or zoning policies that change how one driver’s trip affects congestion for others.

The theoretical literature has extensively studied these two levers, but often in isolation. A seminal contribution by \cite{bcz} identifies the optimal individual whose removal maximally reduces aggregate activity in network games, thus offering guidance on optimal network-based interventions. In contrast, \cite{ggg} analyze how planners should optimally allocate incentives across individuals in the presence of network spillovers, showing that when actions exhibit strategic complements, incentive allocations should be proportional to agents' network centralities. While both types of interventions can be highly effective, they operate through distinct channels. Moreover, interventions along one dimension may affect the effectiveness of interventions along the other. 
These observations motivate the need for a more comprehensive analysis of joint interventions, where both network ties and individual incentives can be shaped simultaneously to achieve welfare-improving outcomes.

This paper develops a model in which a benevolent planner aims to maximize total welfare in a weighted network where agents’ actions depend on both private returns and peer effects, and efforts can be either strategic complements or strategic substitutes. The planner can intervene along two costly dimensions: first, by adjusting agents’ standalone marginal utilities through characteristic interventions (e.g., targeted subsidies or penalties), as in \cite{ggg}; and second, by modifying the strength of links between agents, thereby influencing behavior through peer spillovers, as in \cite{szz}. The central objective is to determine the optimal allocation of a fixed intervention budget across these two instruments.

 In this paper, we take the infrastructure network as our  real-world application, where nodes represent drivers or users. The nature of strategic interaction in this context depends on the externalities at play. Commuting via public transit, using ride-sharing platforms (such as Uber or Lyft), or adopting electric vehicles (EVs) all create situations where usage efforts are \emph{strategic complements}: higher participation improves service frequency, reduces waiting times, or supports denser charging networks, thereby increasing the individual payoff from usage \citep{small2007urban,cachon2017role,hardman2018review}. Accordingly, in terms of standalone incentives, the planner may adjust fare subsidies, offer EV purchase incentives, or provide bonuses for ride-sharing. In terms of link weights, the planner may extend transit lines, install EV charging stations, or design platforms that amplify how one user’s participation benefits others.
By contrast, in the case of road traffic, driving efforts are \emph{strategic substitutes}: when others drive more, congestion lowers the marginal benefit of one’s own driving. Thus, in terms of standalone incentives, the planner may adjust fuel taxes, parking fees, or congestion charges, or provide subsidies for public transport. In terms of link weights, the planner may expand roads, introduce dedicated bus lanes, coordinate traffic lights, or implement zoning policies that determine how strongly one driver’s trips create congestion for others  \citep{wardrop1952road, pigou1920economics, roughgarden2002bad}.
In both cases, the central feature is that the \emph{network structure is exogenous}, determined by physical infrastructure or centrally managed platforms, while individual agents choose only their participation levels.



Our model reveals rich interactions between these intervention channels. The marginal returns to characteristic interventions depend on the strength of peer effects and the connections among agents, while the marginal returns to link interventions depend on agents' network centralities, shaped by both private incentives and network topology. 
In Theorem \ref{prop-1} we derive necessary conditions for optimality  using variational methods. These conditions appear simpler in the characteristic  dimension after applying the spectral decomposition method,  as described in \cite{ggg}. However, the complexity persists because the network dimension introduces more variables to solve, which are intertwined with the characteristic dimension.

Specializing to the case where standalone marginal utilities are negligible, Proposition \ref{th-1} shows that the optimal change in link weights is proportional to the product of the eigen-centralities of the connected agents. These eigen-centralities correspond to the leading eigenvector associated with the largest eigenvalue under strategic complements, and with the smallest eigenvalue under strategic substitutes. Under strategic complements, eigen-centralities have uniform sign, implying that link intensities increase under optimal intervention, with the largest increases concentrated on central agents. Under strategic substitutes, eigen-centralities have mixed signs, naturally partitioning the network into two groups; optimal intervention strengthens cross-group ties while weakening within-group ties.

In Proposition \ref{lem-ext2}, we relax the assumption that standalone marginal utilities are negligible and instead consider the general case. When the budget is sufficiently large, the influence of standalone marginal utilities becomes negligible, and the results from Proposition \ref{th-1} serve as accurate approximations. Furthermore, Proposition \ref{lem-ext2} establishes both lower and upper bounds for the equilibrium utility. These bounds depend on the benchmark case with negligible standalone marginal utilities and the variation in the available budget.
Therefore, the results from Proposition \ref{th-1} give an approximate characterization of the optimal joint intervention problems.  Proposition \ref{lem-ext2} also provides the approximation ratio, which is $1$ plus a term of order $1$ over the square root of the budget. Notably, the approximation ratio approaches $1$ as the budget becomes large.  Furthermore, Proposition  \ref{prop-instance}  demonstrates that the approximation ratio in Proposition \ref{lem-ext2} is asymptotically optimal.

Regardless of the initial structure, Theorem \ref{th-2} establishes that for sufficiently large budgets, the complete network and the balanced complete bipartite network are optimal  under strategic complements and strategic substitutes, respectively.  This result builds on prior observations by \cite{ggg}, showing that the shadow price of the planner's budget is increasing in the leading eigenvalue (for complements) and decreasing in the smallest eigenvalue (for substitutes). Thus, the optimal network design problem reduces to maximizing (or minimizing) the corresponding eigenvalue, as characterized in Lemma \ref{lem-graph}. We further analyze the configuration of the bipartite structure and its computational complexity in Proposition \ref{prop-np}, connecting it to the classical maximum cut problem.

Next, we compare welfare outcomes under joint versus single interventions. Theorem \ref{th-welfare} establishes that network design yields substantial welfare gains, which increase with the strength of spillovers. Theorem \ref{th-ineq} shows that, under sufficiently large budgets, joint interventions can eliminate payoff inequality by equalizing eigen-centralities. In contrast, inequality may persist under single interventions that hold the network structure fixed. However, with moderate budgets, joint interventions may exacerbate inequality due to trade-offs between aggregate efficiency and distributional equity, as illustrated in Example \ref{ex-4}. Proposition \ref{prop-4} quantifies the welfare cost of imposing equality constraints. Finally, we consider several extensions in Propositions \ref{prop-pure}, \ref{prop-util}, and \ref{prop-cost}, showing that while these extensions affect welfare levels, the optimal network structure remains either complete or complete bipartite under large budgets, consistent with Theorem \ref{th-2}.

\subsection{Related Literature}

\subsubsection*{Fixed networks}  

Our paper builds on the linear-quadratic framework introduced by \cite{bcz} and \cite{bkd} to analyze players' activity levels and welfare, contributing to the growing literature on optimal interventions in networks.

A first strand of this literature studies incentive targeting when the network is fixed. In network-based discriminatory pricing, for example, players receive personalized prices depending on their centrality, as shown by \cite{cbo} and \cite{bq}. \cite{dem} extends the analysis to more general targeting frameworks and functional forms, while \cite{boy} examine competitive targeting through advertising and information diffusion, allowing for asymmetric equilibria. Related work includes applications to industrial policy (\citealp{ernest}) and carbon tax reforms via sectoral targeting (\citealp{king}). 

Redistributive policies have also been studied as forms of targeted interventions. In a public good game on a fixed network, \cite{allouch} show that the welfare effects of income redistribution depend on agents' Bonacich centralities. \cite{ggg} study optimal targeting of standalone utilities in networks with strategic complements, using principal component analysis, and demonstrate the importance of eigen-centralities in guiding optimal incentives. Our analysis of joint interventions directly builds on these insights, extending the framework to allow simultaneous intervention on both individual incentives and link weights. While standalone interventions on individual incentives may be relatively easy to implement through pricing or advertising (\citealp{cbo}), modifying the network structure typically requires costly infrastructure or institutional changes with long-term effects (\citealp{oabl}). In such contexts, network interventions become central to welfare maximization. While \cite{ggg} show how principal component analysis guides optimal targeting in fixed networks, our model demonstrates that in the joint intervention setting, eigen-centralities also determine optimal adjustments to link weights.\footnote{Our analysis, which  highlights the role of eigen-centralities in shaping payoff inequality under optimal interventions, is related  to recent work by \cite{elliot}, who shows that Pareto efficiency in public goods networks is linked to the principal eigenvalue of the network,   \cite{penta}, who study implementation problems where robust design depends on the spectral radius of payoff externality networks, and \cite{BFLZ2021} who provide further microfoundations for eigencentralities in network models of perceived competition. } 

A second strand of literature investigates interventions on the network structure itself. Since the seminal work of \cite{bcz} on identifying key players and the subsequent works of \cite{bcz2} and  \cite{gl} on key links, a number of papers have explored optimal network design. These questions are particularly relevant in criminal networks, where interventions target the structure of co-offending relationships, as shown by \cite{mp}. More generally, \cite{bbd} characterize optimal unweighted undirected networks as nested split graphs, while \cite{li} extend these results to weighted and directed networks, identifying generalized nested split graphs as optimal structures.   \cite{svm} provide empirical evidence supporting the welfare relevance of structural interventions in collaborative networks.

\subsubsection*{Endogenous Networks}

A related literature considers models where the network structure is endogenously determined by players' decisions. In particular, \cite{ccz} analyze joint determination of socialization and activity levels, while  \cite{ktz}  and \cite{sg} study endogenous networks with  nested split graphs as equilibrium outcomes. \cite{rogers} compare decentralized equilibrium networks to socially efficient networks, showing conditions under which private and social incentives coincide. \cite{baumann} obtain similar results for reciprocal equilibria where link investments are symmetric. \cite{blochdutta} characterize stars as efficient and stable networks under weighted link formation, while \cite{kinateder} identify complete core-periphery networks as equilibrium outcomes in public good games, and \cite{kinateder2} further show that the highest welfare occurs in a star network where the central hub is the sole provider of the public good, financed by income transfers from the peripheral agents.  \cite{sihua} develops a general framework with link substitutability that generates a variety of equilibrium topologies, and \cite{carlson} studies optimal bipartite network design in two-sided platform settings. 

While previous work by \cite{szz} has studied equivalences between characteristic and structural interventions, our contribution lies in analyzing how these two types of interventions interact under a binding budget constraint. In this respect, our model reveals richer interactions than earlier equivalence results. Relatedly, \cite{hpt} analyze joint design of airline networks and pricing, attributing hub-and-spoke structures to traffic economies. While their planner's objective differs, our model similarly combines network design and incentive decisions, but in a broader strategic environment allowing both complements and substitutes. As a result, our optimal networks depart from the hub-and-spoke structure identified in \cite{hpt}.

Taken together, these literatures underscore the importance of both individual incentives and network structure in shaping welfare outcomes. Our contribution lies in unifying these two intervention levers within a common framework, and characterizing the optimal allocation of resources between them.

The remainder of the paper is as follows. Section \ref{sec-2} introduces the model and the definitions and notations used throughout. Section \ref{sec-Applications} provides real-world applications about our model.  Section \ref{sec-3} provides a characterization of the optimal intervention. Section \ref{sec-4} analyzes the resulting welfare and distributional effects and provides a comparison with the literature without structural interventions. Section \ref{sec-5} discusses some generalizations while Section \ref{sec_Con} concludes the paper. The Appendix  contains the proofs that are omitted in the main text.

\section{Model}\label{sec-2}
\subsection{Setup}

Consider a game on a weighted network \(\b g\) over a set of 
players \(\c N=\{1,\cdots,n\}\). Each player \(i\in\c N\) chooses an action \(x_i\in\mathbb R\) and receives payoff 
\begin{equation}\label{eq-pi}\pi_i(x_i;\b x_{-i})=a_ix_i-\frac{1}{2}x_i^2+\phi\sum_{j=1}^ng_{ij}x_ix_j,\end{equation}
where \(a_i\) represents player \(i\)'s standalone marginal utility, $g_{ij}$ denotes the weight of the link between $i$ and $j$,  and \(\phi\) captures the strategic interactions between players.\footnote{See, for instance, \cite{bcz,bkd,ggg}.} The network  and the standalone marginal utilities  \(a_i\) are exogenous to the players. The case \(\phi>0\) corresponds to strategic complements, while the case \(\phi<0\) corresponds to strategic substitutes.  We use the adjacency matrix \(\b g=(g_{ij})_{1\leq i,j\leq n}\) to summarize the network structure.  We suppose that \( \b g\) is symmetric, has no self-loops, and that there exists an exogenous cap \(\bar w>0\) such that \(g_{ij}\in[0,\bar w]\) for all \(i,j\). That is, \(\b g\) lies in the space
\[\mathcal G_n=\{\b g\in\mathbb R^{n\times n}:g_{ij}=g_{ji}\in[0,\bar w]\text{ for all }i,j,\text{ and }g_{kk}=0\text{ for all }k.\}.\]

Let \begin{equation}\b x=\bmat{x_1\\\vdots\\x_n},\ \b a=\bmat{a_1\\\vdots\\a_n},\ \b g=\bmat{g_{11}&\cdots&g_{1n}\\\vdots&\ddots&\vdots\\g_{n1}&\cdots&g_{nn}}.\end{equation}
In this game, \cite{bcz} showed that the players' equilibrium choices of \(\b x^*\)   satisfy
\begin{equation}
 \b x^*(\b a,\b g)=(x^*_1(\b a,\b g),\cdots,x^*_n(\b a,\b g))^T=[\b I-\phi\b g]^{-1}\b a,
 \label{eq-x*}
\end{equation}
subject to the regularity condition
whereby the largest eigenvalue of \(\phi\b g\) is less than \(1\).\footnote{ This regularity condition guarantees the existence and uniqueness of an equilibrium; see \cite{bcz}.}
 
We will later show in Remark \ref{rmk-2} that this regularity condition is satisfied for \emph{any} $\b g\in\mathcal G_n$ if and only if the following  holds: 

\begin{assumption}\label{as-1}
\[\bar w<\begin{cases}\frac{1}{\phi(n-1)},& when \, \, \phi>0;\\-\frac{2}{\phi n},& when \, \, \phi<0\text{ and }n\text{ is even};\\-\frac{2}{\phi\sqrt{n^2-1}},& when \, \, \phi<0\text{ and }n\text{ is odd}.\end{cases}\]
\end{assumption}

From \eqref{eq-pi}, each player's equilibrium payoff,  as a function of \(\b a\) and \(\b g\), is given by
 \begin{equation}
  \label{eq-pi*} 
 \pi_i(\b x^*(\b a,\b g))=(1/2) (x^*_i(\b a,\b g))^2, i\in N,
 \end{equation}
so (twice of) the total payoff is 
\begin{equation}\label{eq-v}V(\b a,\b g):=2\sum_{i=1}^n\pi_i(\b x^*(\b a,\b g))
=\sum_{i=1}^{n} (x^*_i(\b a,\b g))^2
 =\b a^T[\b I-\phi\b g]^{-2}\b a,\end{equation}
where we use \eqref{eq-pi*} in the second equality and \eqref{eq-x*} in the last equality.

Suppose that the original standalone marginal utilities and network link weights are given by \(\hb a\) and \(\hb g\). The planner is able to intervene on \(\hb a\) and \(\hb g\), and selects post-intervention utilities and network so that  \(\b a\) and \(\b g\)  maximize the players' total payoff \eqref{eq-v}.
Furthermore, we assume that this intervention comes at a quadratic cost to the planner, so the planner solves the system:
\begin{align*}
\max_{\b a\in\mathbb R^n,\ \b g\in\mathcal G_n}\quad& V(\b a,\b g;\hb g,\hb a, C)=\b a^T[\b I-\phi\b g]^{-2}\b a,\\\text{s.t.}\quad&\kappa\|\b g-\hb g\|^2+\|\b a-\hb a\|^2\leq C.\stepcounter{equation}\tag{\theequation}\label{eq-prob}
\end{align*}
Note that we allow $\hat a_i$ and $a_i$ to be negative; in this case, we can interpret $a_i$ as the price or marginal cost of consuming the activity. $C>0$ is the total budget and  \(\kappa\in(0, +\infty]\) is a parameter that measures the relative cost of intervening in \(\b g\) compared to \(\b a\).    
 The quadratic form of the intervention cost greatly simplifies computation, although we expect that qualitatively similar results hold with alternative convex costs. See  Section \ref{sec-5.1} for details of results under alternative specifications of cost functions and objective functions.\footnote{See \cite{ggg} for a related discussion.}

In the special case for  which \(\kappa=+\infty\), we recover the setting of \cite{ggg} where the planner cannot intervene in the network design; thus $\b g=\hb g$. Formally, the planner solves the problem \begin{align*}
\max_{\b a\in\mathbb R^n}\quad& V(\b a,\b g;\hb g,\hb a, C)=\b a^T[\b I-\phi\b g]^{-2}\b a\\\text{s.t.}\quad&\|\b a-\hb a\|^2\leq C, \mbox{ and } \b g=\hb g.\stepcounter{equation}\tag{\theequation}\label{eq-prob2}
\end{align*}
For any finite \(\kappa\), we will refer to the intervention with exogenous \(\b g\) in \eqref{eq-prob2} as the \emph{single intervention} and the intervention with endogenous \(\b g\) in \eqref{eq-prob}  as the \emph{joint intervention}. Consequently, we write the solution to \eqref{eq-prob2} as \(V^*_{single}(\hb g,\hb a,C)\) and the solution to \eqref{eq-prob} as \(V^*_{joint}(\hb g,\hb a,C)\).\footnote{In Section \ref{sec-5.2} we discuss another special case in which $\hb a$ is fixed and the planner can design $\b g$ optimally. }    While problem \eqref{eq-prob} tends to problem \eqref{eq-prob2} in the limit as $\kappa \rightarrow +\infty$, their solutions differ in general when \(\kappa\) is finite.
In addition, we emphasize that the dimension of the joint intervention problem is $n^2$, which is quadratic in the size of the network, whereas in the single intervention problem  \citep{ggg}, the number of variables is $n$, which grows linearly with the size of the network.

\subsection{Notations}
In this paper, for any \(p,q\in\mathbb Z^+\), we write \(\b 1_p\) as the length \(p\) vector of ones, \(\b I_p\) as the \(p\times p\) identity matrix, \(\b J_{pq}\) as the \(p\times q\) matrix of ones, and \(\b 0_p\) as the \(p\times p\) matrix of zeros. If subscripts are omitted, we assume the matrices to be of size \(n\times n\).
We denote  \(K_p\) as the complete graph represented by the adjacency matrix \(\b J_{pp}-\b I_p,\) and \(K_{p,q}\) as the complete bipartite graph represented by the adjacency matrix \(\pmat{\b 0_p&\b J_{pq}\\\b J_{qp}&\b 0_q}.\)

Finally, for any \(p\times p\) symmetric matrix \(\b m\), denote  \(\lmax(\b m)\) and \(\lambda_p(\b m)\) as the largest and smallest eigenvalues of \(\b m\), respectively. Denote further  \(\b u^1(\b m)\) and \(\b u^p(\b m)\) as the representative unit eigenvectors corresponding to \(\lmax(\b m)\) and \(\lambda_p(\b m)\), respectively.\footnote{Pick the eigenvector arbitrarily if \(\lambda_1(\b m)\) or \(\lambda_p(\b m)\) occur with multiplicity larger than 1.}

\section{Real-world applications}\label{sec-Applications}

To illustrate how the model applies in real-world contexts, it is useful to map each component of the model to concrete settings where the planner plausibly controls both the standalone utilities $a_i$ and the link weights $g_{ij}$, while agents choose only their effort levels $x_i$. We propose two leading applications, one for the case of strategic complements ($\phi>0$) and one for strategic substitutes ($\phi<0$).

\subsection{Example of strategic complements: Infrastructure Networks}
\begin{itemize}
    \item \textbf{Nodes ($i$):} Individual users of a transportation infrastructure, such as commuters in a public transit system, ride-sharing drivers, or electric vehicle (EV) owners. 
    \item \textbf{Choice ($x_i$):} Usage effort, such as frequency of taking public transport, intensity of offering rides on a platform (such as Uber or  Lyft), or the extent of EV driving. 
    \item \textbf{Payoff ($\pi_i$):} Each user benefits directly from their own usage ($a_i x_i - \tfrac{1}{2}x_i^2$) and from the participation of others. With $\phi > 0$, efforts are \emph{strategic complements}: higher usage by others raises the payoff from one’s own usage, for instance because more riders support higher service frequency, more EVs justify greater charging infrastructure, or more ride-share drivers reduce waiting times. 
    \item \textbf{Planner’s instruments:}
    \begin{itemize}
        \item Standalone incentives $a_i$: fare subsidies, EV purchase incentives, or bonuses for ride-sharing participation. 
        \item Link weights $g_{ij}$: expansions of transit lines, deployment of charging stations for EV, or platform design that strengthens how strongly one user’s participation benefits others. 
    \end{itemize}
    \item \textbf{Why links are exogenous for the agents?} The infrastructure network (transit lines, charging station locations, or ride-matching algorithms) is centrally planned and exogenous to individual users. Agents can only decide how much to use the system, not how the network is structured. 
\end{itemize}

This setting naturally illustrates \emph{strategic complements} ($\phi>0$), where the utility of using the infrastructure rises with the intensity of others’ usage. Examples include public transit demand increasing with ridership density \citep{small2007urban}, network effects in ride-sharing platforms \citep{cachon2017role}, and EV adoption being reinforced by charging infrastructure and peer usage \citep{hardman2018review}.
	
\subsection{Example of strategic substitutes: Infrastructure Networks}

\begin{itemize}
    \item \textbf{Nodes ($i$):} Individual drivers (or commuter groups) using a city’s transportation system. 
    \item \textbf{Choice ($x_i$):} Driving effort, such as the number of trips taken, distance traveled, or commuting frequency. 
    \item \textbf{Payoff ($\pi_i$):} Drivers benefit from mobility (reaching destinations, accessing jobs and services), modeled by $a_i x_i - \tfrac{1}{2}x_i^2$. Congestion externalities enter through $\phi \sum_j g_{ij} x_i x_j$: with $\phi < 0$, efforts are \emph{strategic substitutes}, since higher driving by others lowers the payoff from one’s own driving through increased congestion. 
    \item \textbf{Planner’s instruments:}
    \begin{itemize}
        \item Standalone incentives $a_i$: fuel taxes, parking fees, congestion charges, or subsidies for public transport. 
        \item Link weights $g_{ij}$: road expansions, dedicated bus lanes, traffic-light coordination, or zoning policies that affect how strongly one driver’s trips generate congestion for others.\footnote{In this case,  the planner can  impose negative marginal utilities ($a_i<0$) to reduce congestion.} 
    \end{itemize}
    \item \textbf{Why links are exogenous for the agents?} The set of interactions between drivers is dictated by the physical road network. Drivers cannot choose who they share road segments with; these interdependencies ($g_{ij}$) are determined by infrastructure and city planning, making the network exogenous. 
\end{itemize}

This formulation provides a clear example of \emph{strategic substitutes} ($\phi<0$), since congestion implies that more driving by others reduces the marginal return to one’s own trips \citep{wardrop1952road,pigou1920economics,roughgarden2002bad}. The same reasoning extends to other infrastructure settings, such as electricity grids, water distribution, or pipeline networks, where agents choose usage but not the underlying connectivity.

\section{Analysis}\label{sec-3}

In this section, we provide two complementary approaches to characterize the planner's program \eqref{eq-prob}.  The first approach uses standard variational analysis to pin down the necessary optimality conditions for any candidate solution. In the second approach, we reformulate program \eqref{eq-prob} as a two-stage program, in which, in the first stage, the planner implements a post-intervention network $\b g$, and, then, in the second stage, selects the optimal post-intervention standalone marginal utilities $\b a$ subject to the adjusted budget (after subtracting the intervening cost of implementing $\b g$). Exploiting several key results in \cite{ggg} in the second stage regarding the optimal $\b a^*$ with an exogenous network $\b g$ and the shadow price of the budget,\footnote{When $\b g$ is given, the problem in the second step is precisely the optimal targeting intervention problem as in \cite{ggg}.} we are able to gain insights into the optimal network endogenously selected by the planner in the first stage.

\subsection{A variational approach}

To obtain the optimal intervention, we first want to determine the marginal increase in the players' total payoff from interventions in both the standalone marginal utilities and the network. Define \[\theta\triangleq\frac{\partial V}{\partial \b a}\text{ and }\xi\triangleq\frac{\partial V}{\partial \b g}\] as the marginal benefits of intervening in $\b a$ and \(\b g\) respectively. Further define the matrix \(\b M=[\b I-\phi\b g]^{-1}\) as in \cite{bcz}. 
\begin{lemma}\label{lem-mb} The marginal benefits \(\theta\) and \(\xi\) are given by the following equations:
\begin{align}
    &\theta=2\b M^2\b a,\label{eq-theta}
    \\&\xi=\phi\b M\b a\b a^T\b M^2+\phi\b M^2\b a\b a^T\b M.\label{eq-xi}
\end{align}
\end{lemma}

Both expressions are obtained by differentiating  \(V=\b a^T\b M^2\b a\) with respect to $\b a$ and \(\b g\), respectively.  
Observe that \eqref{eq-theta} can be rewritten using the equilibrium condition \eqref{eq-x*} as \(\theta=\b M\b x\). Therefore, for any \(i\), the marginal benefit of increasing player \(i\)'s utility is 
\[\theta_i=\sum_{k=1}^n m_{ki}x_k=\sum_{k=1}^n m_{ki}b_k(\b g,\b a),\] 
where \(b_k(\b g,\b a)\) represents the Katz-Bonacich centrality of player \(k\) in   network \(\b g\) with weights \(\b a\). The marginal benefits can thus be seen as a weighted sum of the Katz-Bonacich centralities across the network.

To simplify the analysis, we follow the methods proposed in \cite{ggg},  which decomposes the intervention \(\b a\) into orthogonal principal components of \(\b g\) that  are determined by the network and are ordered according to their associated eigenvalues. Let \(\lambda_1>\cdots>\lambda_n\) be the eigenvalues of \(\b g\),\footnote{We make the generic assumption that the eigenvalues of \(\b g\) are distinct.} and let \(\{\b u_1,\cdots,\b u_n\}\) be an orthonormal basis of \(\mathbb R^n\) such that each \(\b u_k\) is an eigenvector of \(\b g\) with corresponding eigenvalue \(\lambda_k\). Then there exists unique scalars \(\rho_1,\cdots,\rho_n\) such that \[\b a=\sum_{k=1}^n \rho_k\b u_k.\]
Alternatively, let $\rho_k\triangleq \b a^T \b u_k.$ 
Using this decomposition, equation \eqref{eq-theta} becomes 
\[\theta=2\b M^2\sum_{k=1}^n \rho_k\b u_k=2\sum_{k=1}^n\frac{\rho_k}{(1-\phi\lambda_k)^2}\b u_k.\] 
Similarly, equation \eqref{eq-xi} can be written as
\[\xi=2\phi\sum_{k,l=1}^n\frac{\rho_k}{1-\phi\lambda_k}\frac{\rho_l}{(1-\phi\lambda_l)^2}\b u_k\b u^T_l.\]

Finally, the optimal intervention can be determined by equating these marginal benefits with the respective marginal costs of the intervention. Letting \(\mu\) represent the shadow price of the budget, that is, \(\mu=\frac{\partial V^*}{\partial C}\), we can compute the marginal cost of intervention at \(\b a\) to be \(2\mu(\b a-\hb a)\), while the corresponding marginal cost of intervention at \(\b g\) is \(2\mu\kappa(\b g-\hb g)\). However, we have the constraint that \(\b g\) does not have self-loops, so our first order constraint only holds for the off-diagonal entries of \(\b g\). 

By summarizing the principal component analysis above, we can  write the original standalone marginal utilities as \(\hb a=\sum_{k=1}^n\hat\rho_k\b u_k\). Note that \(\mu\) appears as the marginal costs to capture the trade-off in allocating the budget between intervening in \(\b a\) and  \(\b g\). The conditions for optimality are thus summarized as follows:

\begin{theorem}\label{prop-1} The solution to the system \eqref{eq-prob} must satisfy
\begin{align}
    &\label{eq-foc1}\frac{\rho_k}{(1-\phi\lambda_k)^2}\b u_k=\mu(\rho_k-\hat\rho_k)\b u_k, \, \,  \text{ for all }k, \tag{A1}\\
    &\label{eq-foc2}\sum_{k,l=1}^n\frac{\rho_k}{1-\phi\lambda_k}\frac{\rho_l}{(1-\phi\lambda_l)^2}\phi(\b u_k\b u^T_l)_{ij}\left\{\begin{matrix}
        =\mu\kappa(\b g^*-\hb g)_{ij}, & g_{ij}^*\in(0,\bar w); \\
        \leq \mu\kappa(\b g^*-\hb g)_{ij}, & g_{ij}^*=0; \\
        \geq \mu\kappa(\b g^*-\hb g)_{ij}, & g_{ij}^*=\bar w, 
    \end{matrix}
    \right. \text{ for all $i\neq j$,} \tag{A2}\\
    &\label{eq-foc3}\sum_{k=1}^n(\rho_k-\hat\rho_k)^2+\kappa\|\b g^*-\hb g\|^2=C,\tag{A3}
\end{align}
where \((\b u_k,\lambda_k)\) are the eigenpairs of \(\b g^*\), in decreasing order of eigenvalues, while \(\rho_k,\hat \rho_k\) are the magnitudes of the projections of \(\b a^*\) and \(\hb a\) to \(\b u_k\).
\end{theorem}

The first two equations, \eqref{eq-foc1} and \eqref{eq-foc2}, are the first-order conditions with respect to \(\b a\) and \(\b g\),   while the third equation, \eqref{eq-foc3}, is the first-order condition associated with the budget constraint, which must bind at the optimal intervention. These conditions are stated in the eigen-space of $\b g$.

More generally, this theorem characterizes the optimal intervention using a variational approach. The planner must allocate a fixed budget between modifying standalone marginal utilities and altering the network structure. The optimal solution aligns the direction of intervention with the principal components (eigenvectors) of the network. Intuitively, the planner allocates resources to the directions in which the network most effectively amplifies individual incentives, depending on whether strategic interactions are complements or substitutes.

\subsection{The optimal intervention}

We begin by considering the case where \(\hb a=0\), or equivalently, \(\hat\rho_k=0\) for all \(k\). Equation  \eqref{eq-foc1} reduces to 
\[\rho_k\left(\frac{1}{(1-\phi\lambda_k)^2}-\mu\right)=0,\]
for all \(k\). Clearly, this implies that \(\rho_k\neq0\) for at most one value of \(k\). Following \cite{ggg}, the maximizer occurs when \(\rho_1\neq0\) when \(\phi>0\), and at \(\rho_n\neq0\) when \(\phi<0\). Correspondingly, the shadow price of the budget will be \[\mu=\begin{cases}\frac{1}{(1-\phi\lambda_1)^2},&\phi>0;\\\frac{1}{(1-\phi\lambda_n)^2},&\phi<0.\end{cases}\]

Using the above, we can now fix \(\b a\) to be in the direction of \(\b u_1(\phi\b g)\), and maximize the total payoff over all choices of \(\b g\). That is, we write the value function as \begin{equation}\label{eq-f}f(C;\hb g,\kappa)\triangleq V^*_{joint}(\hb g,\hb a=0,C)=\sup_{\b g}\frac{C-\kappa\|\b g-\hb g\|^2}{(1-\lambda_1(\phi\b g))^2}.\end{equation}

Note that for fixed \(\b g\), the expression \(\frac{C-\kappa\|\b g-\hb g\|^2}{(1-\lambda_1(\phi\b g))^2}\) is linear in \(C\). Hence \(f\) is the supremum of a set of linear functions, so \(f\) itself is convex in \(C\). 
Furthermore, the tangent of \(f\) at \(C\), by Envelope Theorem,  is equal to the shadow price of the budget, given by \(\mu^*=(1-\lambda_1(\phi\b g^*))^{-2}\). Here \(\mu^*\) is increasing in \(\lambda_1(\phi\b g^*)\). Therefore, we also obtain that \(\mu^*\) and \(\lambda_1(\phi\b g^*)\) are increasing in \(C\).

To determine \(\b g^*\), we now consider the first order condition with respect to \(\b g\). When \(\phi>0\) and \(g^*_{ij}\in(0,\bar{w})\),\footnote{The case \(\phi<0\) is similar.} we have only \(\rho_1\neq0\), so condition \eqref{eq-foc2} in Theorem \ref{prop-1} reduces to \[\frac{\phi\rho_1^2}{(1-\phi\lambda_1)^3}(\b u_1\b u_1^T)_{ij}=\mu\kappa(\b g^*-\hb g)_{ij}=\frac{1}{(1-\phi\lambda_1)^2}\kappa(\b g^*-\hb g)_{ij}.\]

That is, the optimal intervention on the network structure is proportional to the outer product of the first eigenvector with itself. Consequently, the degree of intervention is larger for links between nodes of high eigen-centrality. Furthermore, note that \(\rho_1^2=\|\b a^*\|^2=C-\kappa\|\b g^*-\hb g\|^2\) by condition \eqref{eq-foc3} in Theorem \ref{prop-1}. We summarize the above results in the following proposition:

\begin{proposition}\label{th-1}
    Suppose \(\hb a=0\) and Assumption \ref{as-1} holds. 

(a) \(\b a^*\) is in the direction of \(\b u_1(\phi\b g^*)\).
    
(b) \(V^*_{joint}(\hb g,\hb a,C)\) is convex in \(C\).

(c) \(\mu^*\) and \(\lambda_1(\phi\b g^*)\) are increasing in \(C\). 

Denote \(\b u_1(\phi\b g):=(u^1_1,\cdots,u^1_n)\). If the solution \(\b g^*\) is interior, then 

(d) \begin{equation}\label{eq-th1d}
        g^*_{ij}-\hat g_{ij}=\frac{\phi(C-\|\b g^*-\hb g\|^2)}{\kappa(1-\lambda_1(\phi\b g))}u^1_iu^1_j\, \, \text{ for all }i\neq j,
    \end{equation}
and \\
(e) \begin{equation}\kappa\|\b g^*-\hb g\|^2=\frac{\phi^2(C-\kappa\|\b g^*-\hb g\|^2)^2}{\kappa(1-\lambda_1(\phi\b g^*))^2}\left(1-\sum_{i=1}^n(u^1_i)^4\right).\label{eq-th1e}\end{equation}

\end{proposition}

When the  standalone marginal utilities are negligible, the planner allocates the entire budget to modifying the network. The first three results, (a), (b), and (c), are important technical findings. Result (d) is more intuitive and interesting.   When all marginal standalone utilities are zero, and given a fixed intervention budget for $\b a$ and network structure $\b g$, \cite{ggg} show that the optimal $\b a^*$ aligns with the direction of $\b u_1(\phi \b g)$. Therefore, the planner's problem becomes:
\[
\sup_{\b g} \frac{C - \kappa\|\b g - \hb g\|^2}{(1 - \lambda_1(\phi \b g))^2}.
\]
Thus, the optimal solution must satisfy the first-order condition (FOC), which for any $i \neq j$ is given by:
\[
\underbrace{-\kappa(g_{ij} - \hat{g}_{ij})}_{\text{costs of changing } g_{ij}} + \underbrace{(C - \kappa\|\b g - \hb g\|^2)\frac{\phi u_i^{1} u_{j}^{1}}{\kappa(1 - \lambda_1(\phi \b g))}}_{\text{benefits from increasing } \lambda_1 \text{ by changing } g_{ij}} = 0,
\]
which is exactly what Proposition \ref{th-1}(d) states.

Furthermore, result (d) characterizes the optimal intervention by relating it to the principal components of the network, showing that the degree of intervention in the strength of the link between two players is proportional to the product of their eigenvector weights. In other words, the higher the eigenvector weights of two agents, the greater the weight assigned to them by the planner. We clarify this result by considering the following ratio:
\begin{equation}
\frac{g^*_{ij}-\hat{g}_{ij}}{g^*_{ik}-\hat{g}_{ik}} = \frac{u^1_j}{u^1_k}, \label{eq-th1}
\end{equation}
for any \(g^*_{ij}, g^*_{ik} \in (0, \bar{w})\) such that \(u^1_k \neq 0\). Equation \eqref{eq-th1} shows that the ratio of interventions depends solely on the relative components of the eigenvector of \(\b g^*\).

Since the assumption \(\hat{\b a} = \b 0\) in Proposition \ref{th-1} is somewhat restrictive, we now consider the case of a general \(\hat{\b a}\). When the budget \(C\) is sufficiently large, \(\hat{\b a}\) plays a diminishing role, and the results in Proposition \ref{th-1} hold approximately. Moreover, by applying the triangle inequality, we can bound and approximate \(V^*(\hat{\b g}, \hat{\b a}, C)\) by analyzing the case where \(\hat{\b a} = 0\).

\begin{proposition}\label{lem-ext2} Fix \(\hb a,\hb g\), and \(\kappa\). If $C\geq \|\hb a\|^2$,
\begin{align}
V^*(\hb g,0,(\sqrt{C}-\|\hb a\|)^2)\leq V^*(\hb g,\hb a,C)\leq V^*(\hb g,0,(\sqrt{C}+\|\hb a\|)^2).
\end{align}
Moreover, the optimal solution to the problem \(\max V(\b a,\b g; \hb g,\hb a = 0,(\sqrt{C}-\|\hb a\|)^2)\) is a feasible  intervention for the problem \(\max V(\b a,\b g;\hb g,\hb a,C)\) and 
$$
\frac{V^*(\hb g,\hb a,C)}{V^*(\hb g,0,(\sqrt{C}-\|\hb a\|)^2)} \leq 1 + \frac{4\sqrt{C}\|\hb a\|}{(\sqrt{C}-\|\hb a\|)^2-\kappa\|\b g^{*}-\hb g\|^2},
$$
where $\b g^*$ is the optimal network to the problem $\max V(\b a,\b g;\hb g, \hb a=0,(\sqrt{C}+\|\hb a\|)^2)$.
\end{proposition}

Proposition \ref{lem-ext2}   shows that even when initial standalone utilities are not zero, the general intervention problem can be well approximated by solving the simplified version with zero initial utilities and adjusted budget. The approximation error becomes negligible as the budget increases. The intuition is that, for large budgets, the planner's optimal action is dominated by the network's structural features rather than initial heterogeneities.

Specifically, Proposition \ref{lem-ext2}  provides lower and upper bounds for the equilibrium utility based on the case \(\hat{\b a} = \b 0\) and a variation in the budget. The intuition behind this result is that, by the triangle inequality, any feasible intervention to the problem 
$\max V(\b a, \b g; \hat{\b g}, \hat{\b a} = 0, (\sqrt{C} - \|\hat{\b a}\|)^2)$ 
(resp. \(\max V(\b a, \b g; \hat{\b g}, \hat{\b a}, C)\)) is also feasible for the problem 
$\max V(\b a, \b g; \hat{\b g}, \hat{\b a}, C)$ 
(resp. \(\max V(\b a, \b g; \hat{\b g}, \hat{\b a} = 0, (\sqrt{C} + \|\hat{\b a}\|)^2)\)). Therefore, the optimal solution to 
\[
\max V(\b a, \b g; \hat{\b g}, \hat{\b a} = 0, (\sqrt{C} - \|\hat{\b a}\|)^2)
\] 
serves as a good approximation to the general problem \(\max V(\b a, \b g; \hat{\b g}, \hat{\b a}, C)\), and the approximation ratio can be bounded by the ratio 
\[
\frac{f((\sqrt{C} + \|\hat{\b a}\|)^2; \hat{\b g}, \kappa)}{f((\sqrt{C} - \|\hat{\b a}\|)^2; \hat{\b g}, \kappa)}.
\]

Proposition \ref{lem-ext2} thus shows that, while deriving a closed-form solution for the general \(V^*\) is challenging, it can be bounded by studying the case \(\hat{\b a} = 0\), with minimal compromise on the budget \(C\), as characterized by Proposition \ref{th-1}. Furthermore, since \(\kappa \|\b g^* - \hat{\b g}\|^2\) is constant, Proposition \ref{lem-ext2} provides a tractable method (with the optimal solution also described in Proposition \ref{lem-ext2}) to approximate the general joint intervention problem within a factor of \(1 + \mathcal{O}\left(\frac{1}{\sqrt{C}}\right)\).\footnote{We write \(f(x) = \mathcal{O}(g(x))\) if there exist constants \(N\) and \(c\) such that for any \(x > N\), \(f(x) \leq c\, g(x)\).} This also yields the convergence rate, illustrating how Proposition \ref{th-1} holds approximately.

The approximation ratio $1+\mathcal{O}(\tfrac{1}{\sqrt{C}})$ is the best that can be achieved when using the joint intervention problem with $\hb a = \b 0$ to approximate the general problem $\max V (\b a,\b g; \hb g, \hb a, C)$. The intuition is that two factors may affect this ratio: (i) the portion of the budget discarded to ensure feasibility, and (ii) heterogeneity in standalone marginal utilities. However, Propositions \ref{lem-ext2} and \ref{prop-instance} show that only the latter matters, which is unavoidable given our approximation philosophy. In fact, Proposition \ref{prop-instance} constructs a case where, even without any discarded budget (whereas our approximation of discarding budget is of order $\mathcal{O}(\sqrt{C})$), the ratio cannot be improved beyond $1+\mathcal{O}(\tfrac{1}{\sqrt{C}})$.

\begin{proposition}\label{prop-instance}
There exist $(\hb g,\hb a)$ such that
\[
    \frac{V^*(\hb g,\hb a,C)}{V^*(\hb g,0,C)} = 1+\Theta\!\left(\frac{1}{\sqrt{C}}\right).\footnote{We write \(f(x) = \Theta(g(x))\) if there exist constants \(N\), \(c_1>0\), and \(c_2>0\) such that for all \(x > N\), \(c_1 g(x)\leq f(x) \leq c_2\, g(x)\).}
\]
\end{proposition}

To illustrate, consider the general problem $\max V (\b a,\b g; \hb g = \bar{w}K_n, \hb a = (1,1,\ldots,1)^{\top}, C)$. Its exact value function is
\[
   V^*(\hb g,\hb a)=\frac{n}{(1-(n-1)\phi\bar{w})^2}\left(1+\sqrt{\tfrac{C}{n}}\right)^2.
\]
In contrast, for the restricted problem $\max V (\b a,\b g; \hb g, \hb a=\b 0, C)$, the exact value is
\[
   V^*(\hb g,\b 0)=\frac{C}{(1-(n-1)\phi\bar{w})^2}.
\]
Thus, the ratio equals
\[
   \frac{n+2\sqrt{nC}+C}{C}=1+\Theta\!\left(\tfrac{1}{\sqrt{C}}\right),
\]
showing that the approximation ratio in Proposition \ref{lem-ext2} is asymptotically optimal.

Next, using the approximation results (Proposition \ref{lem-ext2}), we return to the case $\hb a = \b 0$ and provide insights into the network structure when the budget is sufficiently large. Recall that the eigenvector \(\b u_1(\phi \b g^*)\) corresponds to the largest eigenvalue \(\lambda_{\max}(\b g^*)\) under strategic complementarity (\(\phi > 0\)), and to the smallest eigenvalue \(\lambda_{\min}(\b g^*)\) under strategic substitution (\(\phi < 0\)). Optimal interventions in the network differ dramatically depending on the sign of \(\phi\). We have:

\begin{enumerate}
\item[(i)]
When \(\phi > 0\), by the Perron-Frobenius theorem, the signs of \(u^1_i\) are identical for all \(i\), implying that \(g_{ij}^* > \hat{g}_{ij}\) for all \(i, j\) by \eqref{eq-th1d}. That is, the planner does not reduce the weight of any link. Intuitively, this suggests that, as the budget increases, the optimal graph tends toward the complete graph (we will formalize this observation later).

\item[(ii)]
When \(\phi < 0\), we can partition the players into two subsets\footnote{We ignore nodes with \(u^1_i = 0\), since their weights remain unchanged by Proposition \ref{th-1}.}:
\[
S^+ = \{i : u^1_i > 0\}, \quad S^- = \{j : u^1_j < 0\}.
\]
The planner increases the weights of links across the two sets while reducing the weights of links within each set. Specifically, by \eqref{eq-th1d},
\[
g_{ij}^* - \hat{g}_{ij} =
\begin{cases}
> 0 & \text{if } (i \in S^+,\, j \in S^-) \text{ or } (i \in S^-,\, j \in S^+), \\
< 0 & \text{if } i, j \in S^+ \text{ or } i, j \in S^-.
\end{cases}
\]
\end{enumerate}
We illustrate  the latter  in Figure \ref{fig-2}. 
As the budget grows large, the links across the sets \(S^+\) and \(S^-\) increase to \(\bar w\), while the links within the sets decrease to \(0\). This reveals a tendency for the network to take  a complete bipartite structure as $C$ increases.

\begin{figure}[H]
\begin{center}
\begin{tikzpicture}[
roundnode/.style={circle, draw=black!100, very thick, minimum size=20mm}
]
\node[roundnode] (1) at (0,0) {}; 
\node at (0,-1.5) {\(S^+\)};
\node[roundnode] (2) at (6,0) {}; 
\node at (6,-1.5) {\(S^-\)};
\draw[blue] (1.east)--(2.west) node[midway,above] {\(g_{ik}\uparrow\)};
\draw[red] (1.north west)..controls +(-1.5cm,0.5cm) and +(-1.5cm,-0.5cm) .. (1.south west) node[midway, left]{\(g_{ij}\downarrow\)};
\draw[red] (2.north east)..controls +(1.5cm,0.5cm) and +(1.5cm,-0.5cm) .. (2.south east) node[midway, right]{\(g_{kl}\downarrow\)};
\end{tikzpicture}
\end{center}
\caption{Changes in edge weights}

\label{fig-2}
\end{figure}


\subsection{The case of large budgets}\label{sec-3.3}
We now formally analyze the case where the planner's budget is large. In this regime, joint intervention leads to significant differences. We begin by showing that, in this case, the optimal network always takes a simple form.

\begin{theorem}\label{th-2}
Suppose \(\bar w\) satisfies Assumption \ref{as-1}. 

(a) If \(\phi>0\), then there exists \(\overline C \) such that for all \(C>\overline C\), \[\b g^*(C)=\bar wK_n.\]

(b) If \(\phi<0\), then there exists \( \overline{\overline{C}}\) such that for all \(C>\overline{\overline{C}}\), \[\b g^*(C)\cong\bar wK_{\lfloor\frac{n}{2}\rfloor,\lceil\frac{n}{2}\rceil}.\footnote{Given two graphs \(H\) and \(H'\) on \(p\) vertices, we say \(H\) is isomorphic to \(H'\) \((H\cong H')\) if there exists a permutation \(\sigma\) on \(\{1,\cdots,p\}\) such that \(h_{ij}=h'_{\sigma(i)\sigma(j)}\) for all \(i,j\).}\]

\end{theorem}

Theorem \ref{th-2} identifies the optimal network architecture for large budgets. With strategic complements ($\phi > 0$), the planner benefits from reinforcing mutual connections, leading to a complete network  \(K_n\). With strategic substitutes ($\phi < 0$), the planner minimizes redundant interactions, favoring a complete bipartite network \(K_{\lfloor \frac{n}{2} \rfloor, \lceil \frac{n}{2} \rceil}\). These simple structures are optimal because they maximize or minimize the key spectral value (the largest or smallest eigenvalue), which drives the equilibrium multiplier.
We have:\footnote{We write \(f(x)=o(x)\) if for any \(\epsilon>0\) there exists \(x_0\) such that \(|f(x)|<\epsilon x\) for all \(x>x_0\).} 
\begin{align}V^*_{joint}(\hb g,\hb a,C)&=\max_{\b g\in\mathcal G_n}V_{single}^*(\b g,\hb a,C-\kappa\|\b g-\hb g\|^2).\notag\\
&=\max_{\b g\in\mathcal G_n}\frac{1}{(1-\lmax(\phi\b g))^2}(C-\kappa\|\b g-\hb g\|^2)+o(C).\notag\\
&=\max_{\b g\in\mathcal G_n}\frac{1}{(1-\lmax(\phi\b g))^2}C+o(C).\label{eq-o}
\end{align}
The first equality follows from reinterpreting program \eqref{eq-prob} as a sequential maximization problem, while the last equality holds since the cost of network design, \(\kappa \|\b g - \hat{\b g}\|^2\), is bounded. For sufficiently large \(C\), the dominant term in expression \eqref{eq-o} is governed by the \emph{social multiplier} \(\frac{1}{(1 - \lmax(\phi \b g))^2}\), which is increasing in \(\lambda_1(\phi \b g)\). That is, if \(\b g\) and \(\b g'\) are two networks such that \(\lmax(\phi \b g') > \lmax(\phi \b g)\), then 
\[
V_{\text{single}}^*(\b g', \hat{\b a}, C - \kappa \|\b g' - \hat{\b g}\|^2) > V_{\text{single}}^*(\b g, \hat{\b a}, C - \kappa \|\b g - \hat{\b g}\|^2)
\]
whenever \(C\) is sufficiently large. Consequently, as the budget tends to infinity, the largest eigenvalue of the optimal network under joint intervention, \(\lmax(\phi \b g^*)\), must approach the maximal possible value among all \(\b g \in \mathcal{G}_n\).\footnote{Furthermore, the above argument can be strengthened to show that the optimal network must coincide with either the complete or the complete balanced bipartite graph for sufficiently large \(C\), as stated in the theorem; the technical details are provided in the Appendix.} 

When \(\phi > 0\), we have \(\lambda_1(\phi \b g^*) = \phi \lambda_1(\b g^*)\), while when \(\phi < 0\), we have \(\lambda_1(\phi \b g^*) = \phi \lambda_n(\b g^*)\). Hence, depending on the sign of \(\phi\), we either seek the network maximizing the largest eigenvalue \(\lmax\) or minimizing the smallest eigenvalue \(\lmin\). Lemma \ref{lem-graph} in the Appendix characterizes the eigenvalue-maximizing networks.

The largest eigenvalue of a nonnegative matrix is monotone in its entries (see, for instance, the Perron–Frobenius theorem). Therefore, the largest \(\lambda_1(\b g)\) is achieved when \(\b g\) corresponds to \(\bar{w} K_{n}\). The problem of finding the smallest possible \(\lambda_n\) in the case of unweighted graphs has been studied by \cite{bkd}.\footnote{See also \cite{constantine}.} In particular, \cite{bkd} show that for any unweighted graph \(\b g\) on \(n\) vertices, $\lambda_n(\b g) \geq \lambda_n\left(K_{\lfloor \frac{n}{2} \rfloor, \lceil \frac{n}{2} \rceil}\right).$ 
In Lemma \ref{lem-graph}, we show that a similar argument can be extended to general weighted graphs.
As a by-product, Lemma \ref{lem-graph}  justifies our choice of bounds in Assumption \ref{as-1}. We have:
\begin{remark}\label{rmk-2}
\(\lmax(\phi\b g)<1\) for all \(\b g\in\mathcal G_n\) whenever Assumption \ref{as-1} holds.
\end{remark}

That is, the regularity condition \(\lmax(\phi\b g)<1\) is satisfied for any choice of intervention by the planner, ensuring that the players' equilibrium exists.

To complete our analysis of the optimal joint intervention, we now characterize the optimal choice of \(\b a^*\), which is the main focus of \cite{ggg}. Remark \ref{rmk-3} follows by applying equation \eqref{eq-foc1} to the post-intervention network characterized in Theorem \ref{th-2}, with corresponding eigenvectors \(\b u_1(\phi \b g^*)\) given in Fact \ref{fact-1}.

\begin{fact}\label{fact-1}
\begin{enumerate}
\item[(a)] The largest eigenvalue of \(K_n\) is \(\lmax(K_n) = n-1\), with corresponding eigenspace \(\mathrm{span}\{(1,1,\dots,1)\}\).
\item[(b)] The smallest eigenvalue of \(K_{\lfloor \frac{n}{2} \rfloor, \lceil \frac{n}{2} \rceil}\) is 
\[
\lambda_n\left(K_{\lfloor \frac{n}{2} \rfloor, \lceil \frac{n}{2} \rceil}\right) = -\sqrt{\lfloor \tfrac{n}{2} \rfloor \lceil \tfrac{n}{2} \rceil},
\]
with corresponding eigenspace
\[
\mathrm{span}\{( 
\underbrace{\sqrt{\lceil \tfrac{n}{2} \rceil}, \dots, \sqrt{\lceil \tfrac{n}{2} \rceil}}_{\lfloor \frac{n}{2} \rfloor \text{ terms}},
\underbrace{-\sqrt{\lfloor \tfrac{n}{2} \rfloor}, \dots, -\sqrt{\lfloor \tfrac{n}{2} \rfloor}}_{\lceil \frac{n}{2} \rceil \text{ terms}}
) \}.
\]
\end{enumerate}
\end{fact}
This fact presents spectral properties of the complete and complete bipartite graphs. These structures attain the extreme eigenvalues (largest or smallest), which makes them optimal for maximizing the planner's objective under large budgets. The associated eigenvectors are also simple: uniform for the complete graph, and two-level for the bipartite case.

\begin{remark}\label{rmk-3}
Suppose \(\bar{w}\) satisfies Assumption \ref{as-1}.
\begin{enumerate}
\item[(a)] If \(\phi > 0\), then there exists \(\xi \in \{1, -1\}\) such that
\[
\lim_{C \to \infty} \frac{\b a^*(C) - \hat{\b a}}{\sqrt{C}} = \frac{\xi}{\sqrt{n}} \b 1_n.
\]
\item[(b)] If \(\phi < 0\), then there exists a sequence \((c_i)\) with \(c_i \to \infty\), and a choice of eigenvector \(\b u_n\left(K_{\lfloor \frac{n}{2} \rfloor, \lceil \frac{n}{2} \rceil}\right)\) in the eigenspace of 
\(\lambda_n\left(K_{\lfloor \frac{n}{2} \rfloor, \lceil \frac{n}{2} \rceil}\right) = -\sqrt{\lfloor \tfrac{n}{2} \rfloor \lceil \tfrac{n}{2} \rceil}\),
such that
\[
\lim_{i \to \infty} \frac{\b a^*(c_i) - \hat{\b a}}{\sqrt{c_i}} = \b u_n\left(K_{\lfloor \tfrac{n}{2} \rfloor, \lceil \tfrac{n}{2} \rceil}\right).
\]
\end{enumerate}
\end{remark}

This remark explains how the optimal standalone utilities converge in direction as the budget grows. For $\phi > 0$, all agents are treated symmetrically and receive equal boosts. For $\phi < 0$, the intervention splits the population into two groups with opposite signs, reflecting the optimal bipartition. This aligns with the dominant eigenvector of the corresponding optimal network.

Thus far, we have established that when \(\phi\) is negative, \(\b g^*(C)\) converges to the complete balanced bipartite graph for large \(C\). The remaining issue is to identify the optimal partition into two balanced subsets---that is, subsets of (approximately) equal size. When the standalone marginal utilities are identical, i.e., \(\hat{a}_i = \hat{a}_j\) for all \(i, j\), the optimal partition minimizes the cost of intervention in the network weights, as all nodes are otherwise symmetric. However, we show that this problem is computationally difficult even in this special case.

\begin{proposition}\label{prop-np}
When \(\hat a_i=\hat a_j\) for all \(i,j\), the configuration problem of choosing the optimal partition of \(\b g^*\) that maximizes total payoffs is NP-hard.
\end{proposition}

This result underscores the computational complexity of determining the optimal bipartite partition in the case of strategic substitutes. Even when all agents are initially symmetric, finding the best division is NP-hard, as it corresponds to a constrained maximum cut problem. This highlights a practical limitation for implementing optimal interventions in large systems.

Specifically, given a sufficiently large budget \(C\), Theorem \ref{th-2} tells us that \(\b g^*\) must be isomorphic to \(K_{\lfloor\frac{n}{2}\rfloor,\lceil\frac{n}{2}\rceil}\), with value function \[V_{joint}^*(\hb g,\hb a,C)=\mu\|\b a^*\|^2=\frac{1}{(1-\phi\lmin(K_{\lfloor\frac{n}{2}\rfloor,\lceil\frac{n}{2}\rceil}))^2}(\sqrt{C-\kappa\|\b g^*-\hb g\|^2}+\|\hb a\|)^2.\]
Therefore, \(V^*_{joint}(\hb g,\hb a,C)\) is maximized when \(\kappa\|\b g^*-\hb g\|^2\) is minimized. Letting \(S\) be a part of the partition of \(\mathcal N\) induced by \(\b g^*\) with size \(\left\lfloor\frac{n}{2}\right\rfloor\), we see that \(\b g^*\) minimizes \[\|\b g^*-\hb g\|^2=\sum_{i,j\in S}\hat g_{ij}^2+\sum_{i,j\notin S}\hat g_{ij}^2+\sum_{\substack{i\in S\\j\notin S}}(\bar w-\hat g_{ij})^2=\|\hb g\|^2+2\left\lfloor\frac{n}{2}\right\rfloor\left\lceil\frac{n}{2}\right\rceil\bar w^2-2\bar w\sum_{i\in S, j\notin S}\hat g_{ij}.\]
Recall the definition of the weight of a cut \(S\subset\mathcal N\) as \[Cut(S)=\sum_{i\in S, j\notin S}\hat g_{ij},\]
so that the orientation \(\b g^*\) of \(K_{\lfloor\frac{n}{2}\rfloor,\lceil\frac{n}{2}\rceil}\) is the one that maximizes \(Cut(S)\). Aside from the constraint on the size of \(S\), this is similar to the nonnegative weighted maximum cut problem (MAX-CUT), which is known to be NP-hard \citep{karp}. In the Appendix, we complete the proof of NP-hardness of the orientation problem by showing reducibility from the constrained version. The homogeneous standalone marginal utilities here is a special case of the general joint intervention problem. Therefore, we also show that the general problem \eqref{eq-prob} is at least NP-hard.

\subsection{Simulations for intermediate budgets}

In this section, we present simulation results to illustrate Theorem \ref{th-2}. The original standalone marginal utilities and the initial network are given in the following example. We plot the optimal networks under intermediate budgets for the cases \(\phi > 0\) (\(\phi = 0.05\)) and \(\phi < 0\) (\(\phi = -0.05\)), respectively. All other parameter settings are provided in Example \ref{ex-2}.

\begin{example}\label{ex-2}
Let \(n=8,\kappa=0.25,\ \bar w=1\), and \[\hb a=\pmat{1 \\ 0.2\\0.1667 \\0.1333 \\ 0.1\\ 0.0667\\0.0333\\0}, ~~ \hb g=\pmat{0&0.5&0.5&0.5&0.5&0.5&0.5&0.5 \\
0.5&0&0&0&0&0&0&0\\0.5&0&0&0&0&0&0&0\\  0.5&0&0&0&0&0&0&0\\0.5&0&0&0&0&0&0&0\\0.5&0&0&0&0&0&0&0\\0.5&0&0&0&0&0&0&0\\0.5&0&0&0&0&0&0&0}.\]

\end{example}

First, we present simulation results for the case \(\phi = 0.05\). Figures \ref{fig:sim+0} to \ref{fig:sim+50} display the optimal networks under different budget levels \(C\). In these graphs, the width of each edge reflects its weight, while the color indicates the rate of change relative to the original network. Gray denotes no change, blue indicates a positive change, and red indicates a negative change. The darker the color, the larger the magnitude of the change.

\begin{figure}[htbp]
\begin{subfigure}{0.45\textwidth}
\begin{center}
\includegraphics[scale = 0.25]{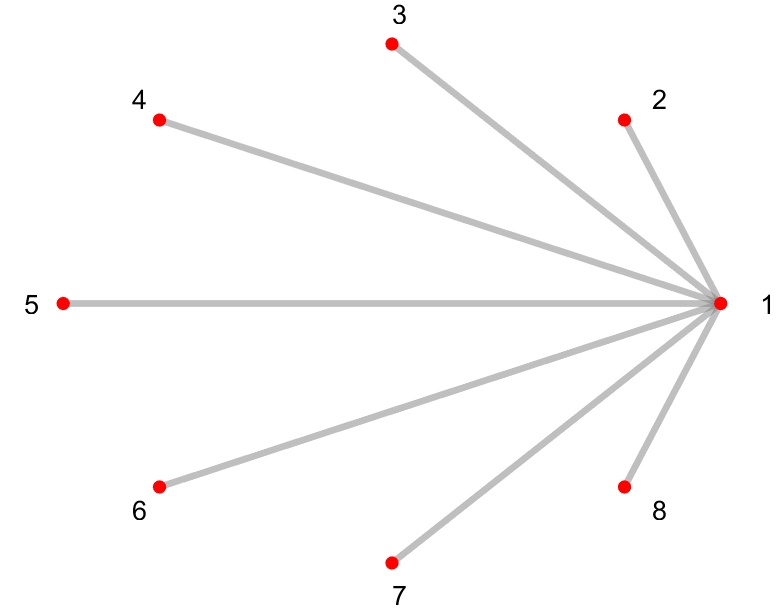}
\end{center}
\caption{\(C=0\)}\label{fig:sim+0}
\end{subfigure}
\qquad
\begin{subfigure}{0.45\textwidth}
\begin{center}
\includegraphics[scale = 0.25]{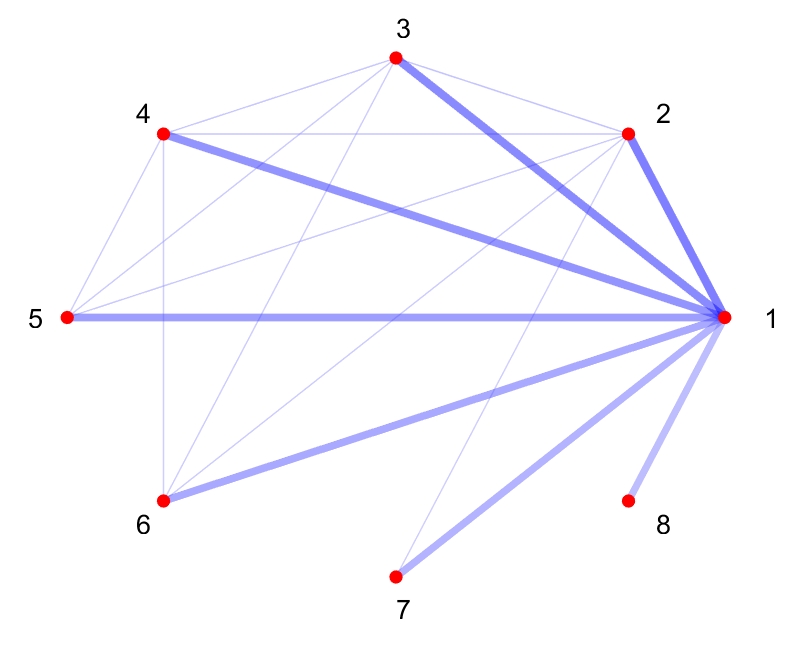}
\end{center}
\caption{\(C=1\)}\label{fig:sim+1}
\end{subfigure}
\\
\begin{subfigure}{0.45\textwidth}
\begin{center}
\includegraphics[scale = 0.25]{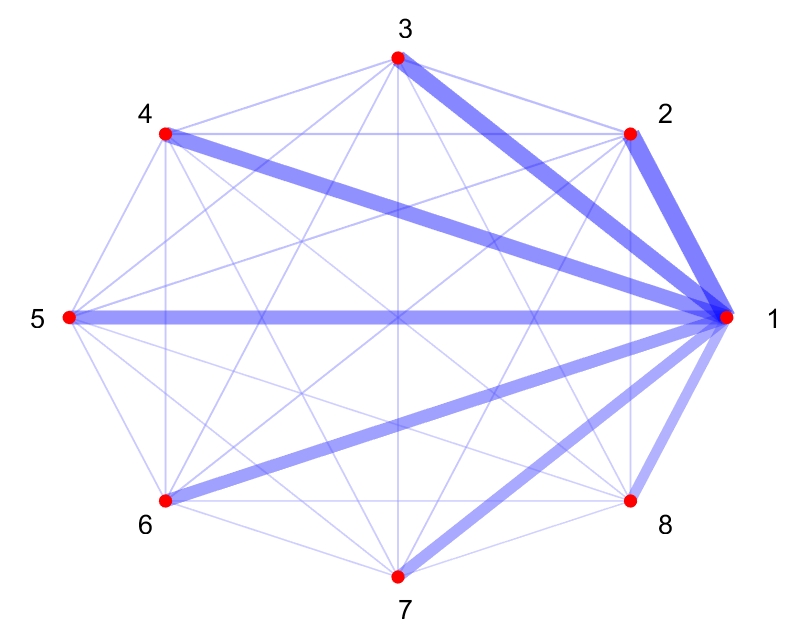}
\end{center}
\caption{\(C=6\)}\label{fig:sim+6}
\end{subfigure}
\qquad
\begin{subfigure}{0.45\textwidth}
\begin{center}
\includegraphics[scale = 0.25]{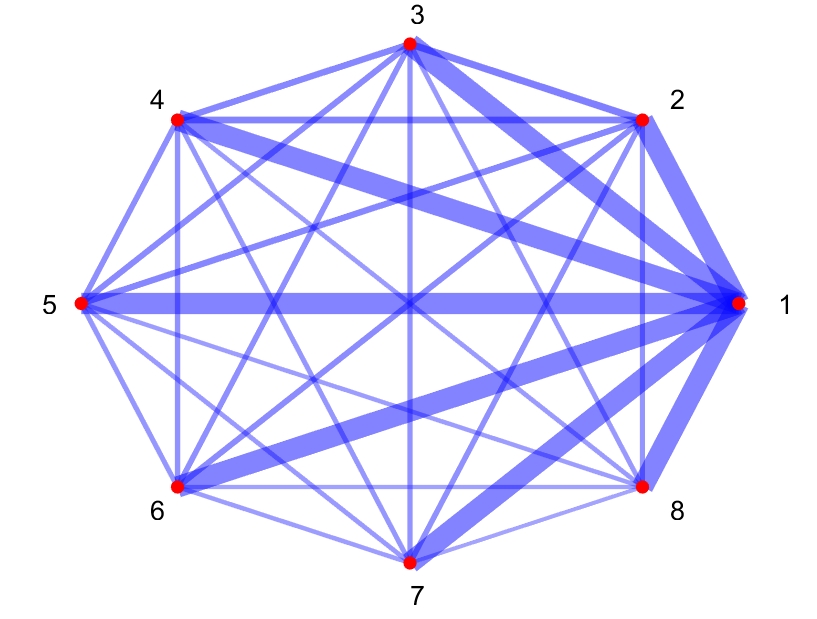}
\end{center}
\caption{\(C=20\)}\label{fig:sim+20}
\end{subfigure}
\\
\begin{subfigure}{0.45\textwidth}
\begin{center}
\includegraphics[scale = 0.25]{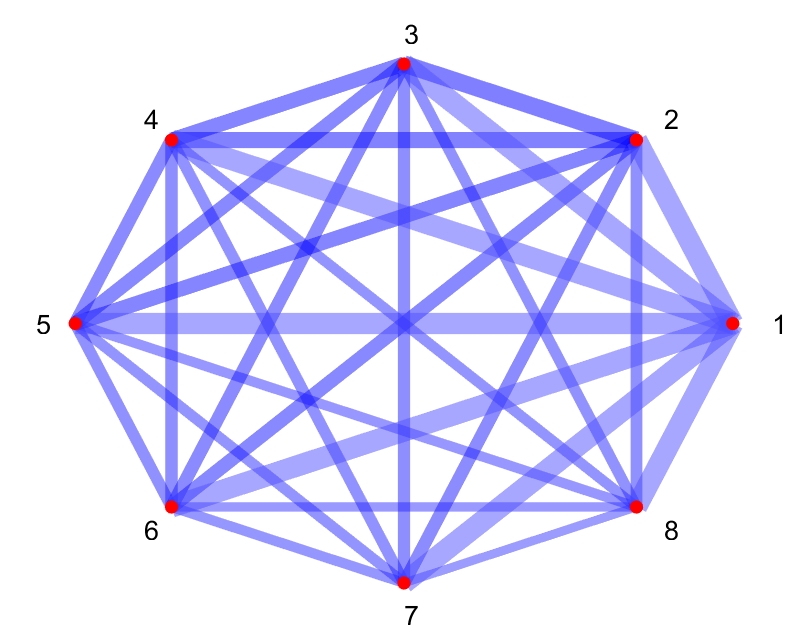}
\end{center}
\caption{\(C=30\)}\label{fig:sim+30}
\end{subfigure}
\qquad
\begin{subfigure}{0.45\textwidth}
\begin{center}
\includegraphics[scale = 0.25]{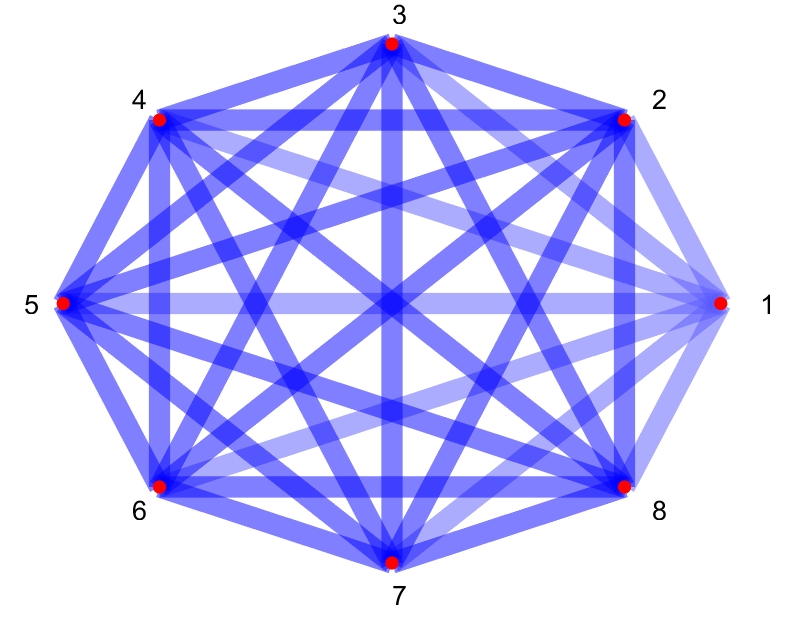}
\end{center}
\caption{\(C\geq 50\)}\label{fig:sim+50}
\end{subfigure}
\caption{Optimal Networks ($\phi=0.05$)}\label{fig:sim+g}
\end{figure}

When \(\phi > 0\), as predicted, the optimal network eventually converges to the complete network \(\bar{w} K_{n}\). In this case, each \(g_{ij}\) increases monotonically until it reaches the upper bound \(\bar{w}\), though at different rates. Figure \ref{fig-sim+a} plots the absolute value of the cosine\footnote{We use the standard notion of cosine similarity: the similarity between two vectors is given by the cosine of the angle between them in the plane they jointly define.} of the angle between \(\b a^*(C)\) and the eigenvector corresponding to the largest eigenvalue of \([I - \phi \b g^*(C)]^{-1}\). As \(C\) increases, \(\b a^*\) and \(\b g^*\) approach \((a^*, \dots, a^*)^{T}\) and \(\bar{w} K_{n}\), respectively. However, the convergence rate remains ambiguous. 

We also plot the evolution of \(\lambda_1(\phi \b g^*)\) with respect to \(C\) in Figure \ref{fig:sim+lambda}. As shown in Figure \ref{fig:sim+lambda}, \(\lambda_1\) increases monotonically, and the optimal network reaches \(\bar{w} K_n\) when \(C\) lies between 35 and 40. 

\begin{figure}[h]
\begin{subfigure}{0.45\textwidth}
\begin{center}
\includegraphics[scale=0.2]{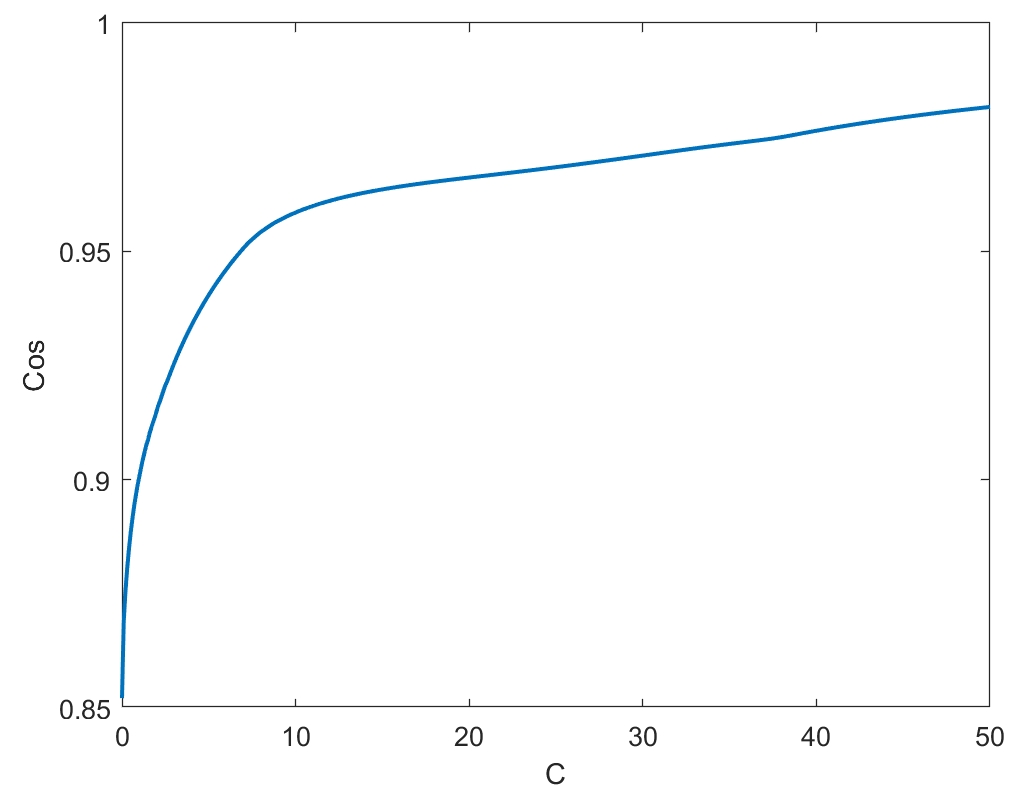}
\end{center}
\caption{The cosine of $\b a^*$ and the eigenvector.}
\label{fig-sim+a}
\end{subfigure}
\qquad
\begin{subfigure}{0.45\textwidth}
\begin{center}
\includegraphics[scale = 0.2]{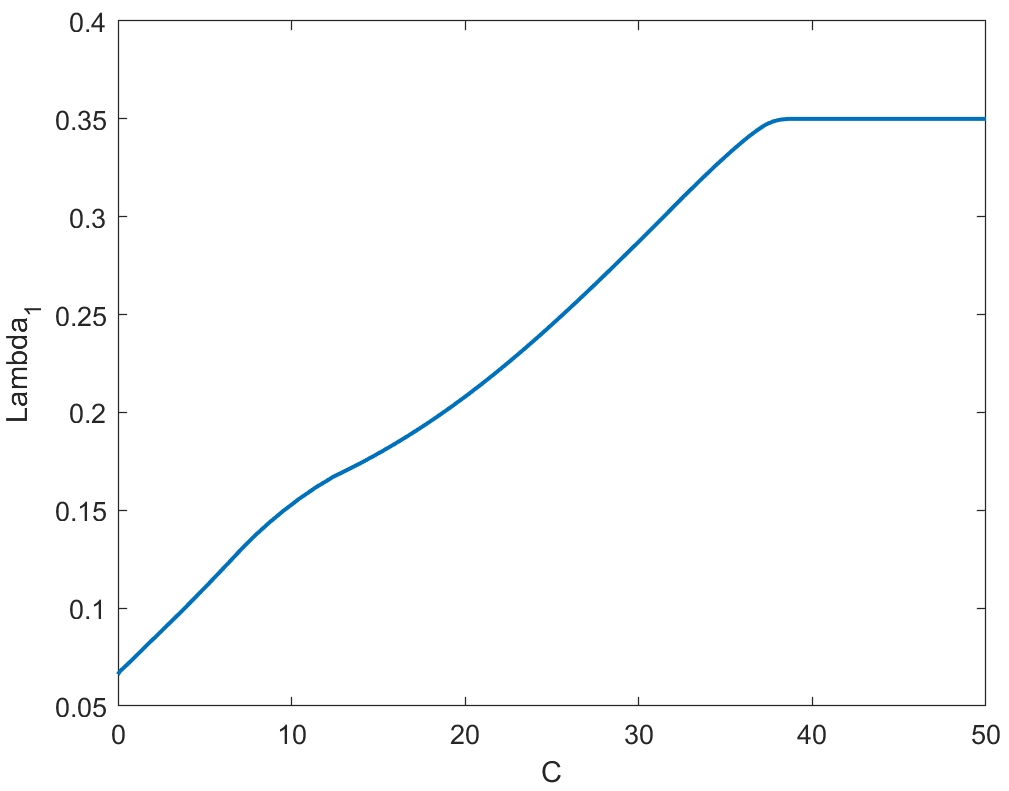}
\end{center}
\caption{\(\lambda_1(\phi\b g^*)\)).($\phi=0.05$)}\label{fig:sim+lambda}
\end{subfigure}
\caption{}
\end{figure}
Second, in Figures \ref{fig:sim-1}-\ref{fig:sim-59}, we present the simulation results and illustrate the ambiguous effect of \(C\) on the optimal network and equilibrium actions when \(\phi = -0.05 < 0\), under intermediate budget levels. 
\begin{figure}[htbp]
\begin{subfigure}{0.45\textwidth}
\begin{center}
\includegraphics[scale = 0.25]{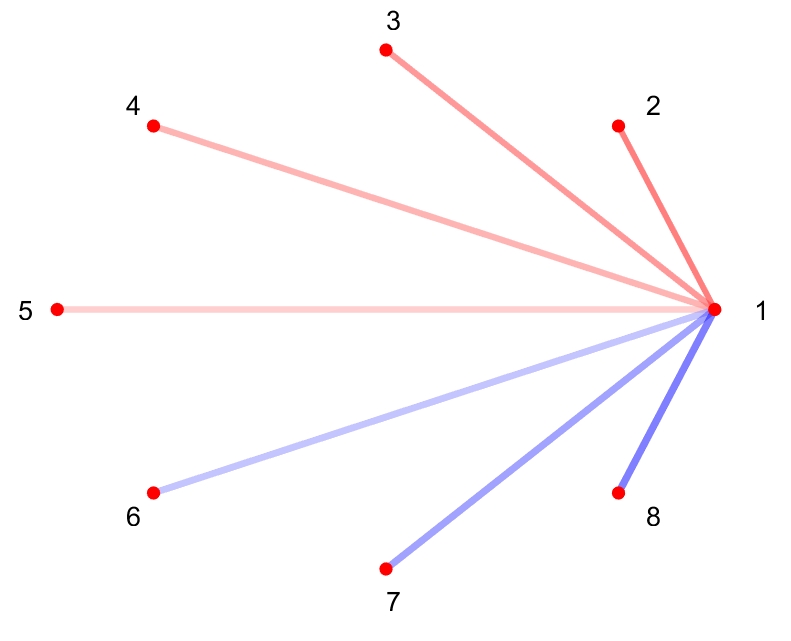}
\end{center}
\caption{\(C=1\)}\label{fig:sim-1}
\end{subfigure}
\qquad
\begin{subfigure}{0.45\textwidth}
\begin{center}
\includegraphics[scale = 0.25]{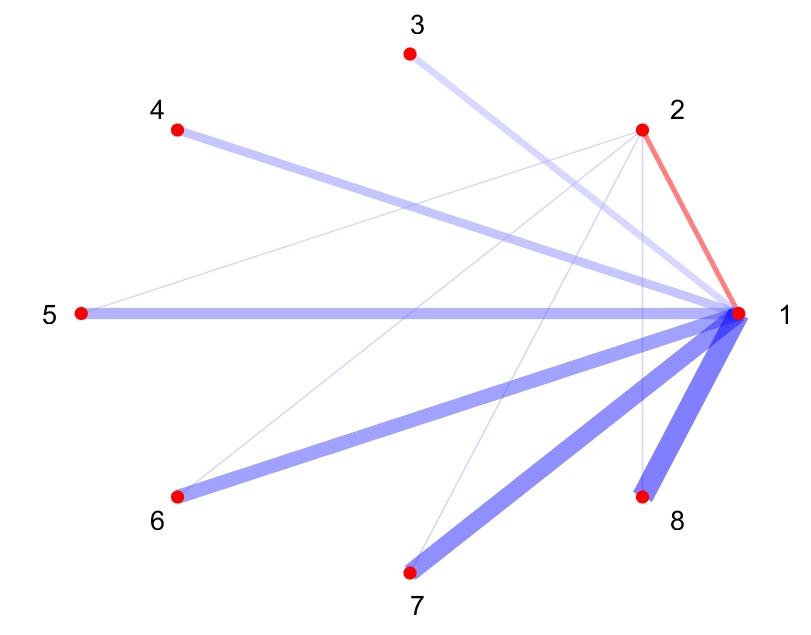}
\end{center}
\caption{\(C=8\)}\label{fig:sim-8}
\end{subfigure}
\\
\begin{subfigure}{0.45\textwidth}
\begin{center}
\includegraphics[scale = 0.25]{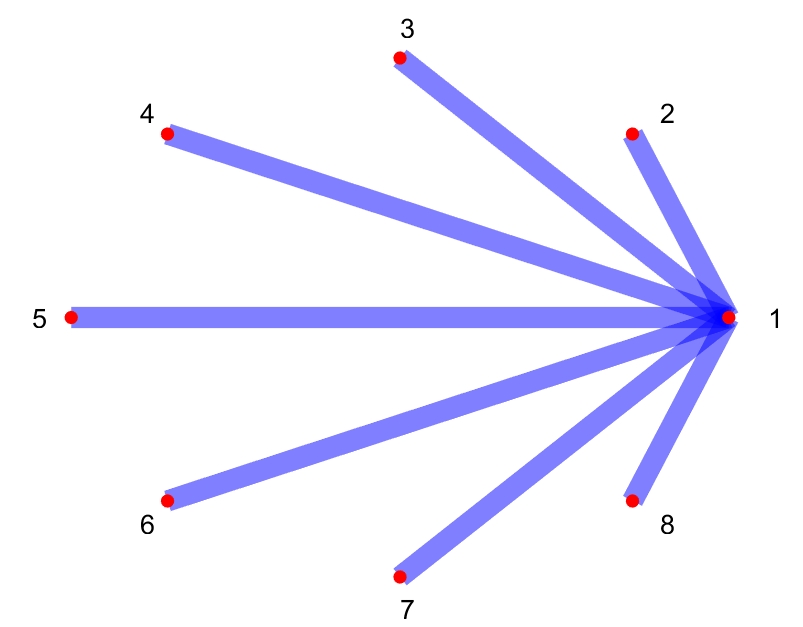}
\end{center}
\caption{\(C=34\)}\label{fig:sim-34}
\end{subfigure}
\qquad
\begin{subfigure}{0.45\textwidth}
\begin{center}
\includegraphics[scale = 0.25]{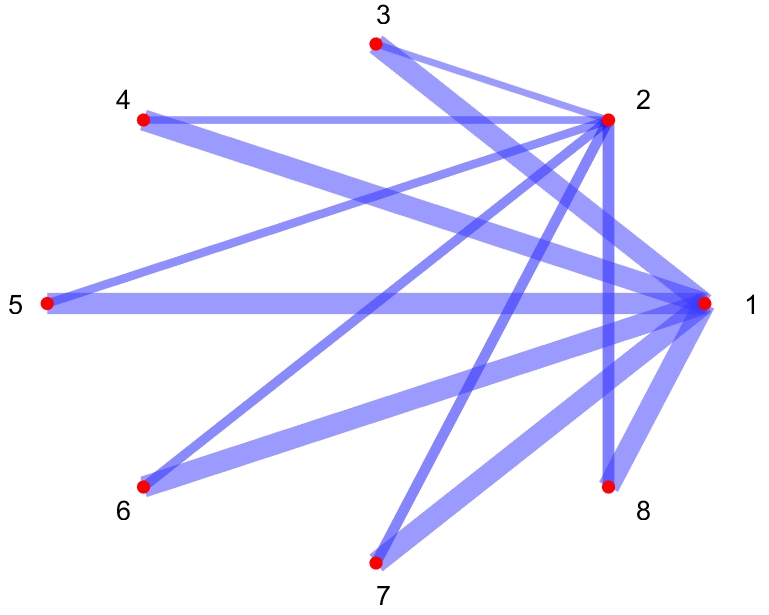}
\end{center}
\caption{\(C=35\)} \label{fig:sim-35}
\end{subfigure}
\\
\begin{subfigure}{0.45\textwidth}
\begin{center}
\includegraphics[scale = 0.25]{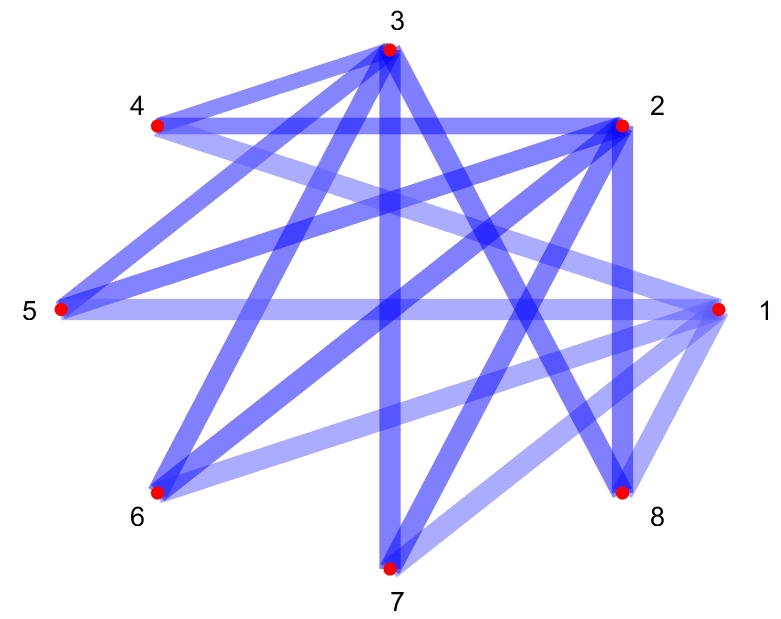}
\end{center}
\caption{\(C=43\)}\label{fig:sim-43}
\end{subfigure}
\qquad
\begin{subfigure}{0.45\textwidth}
\begin{center}
\includegraphics[scale = 0.25]{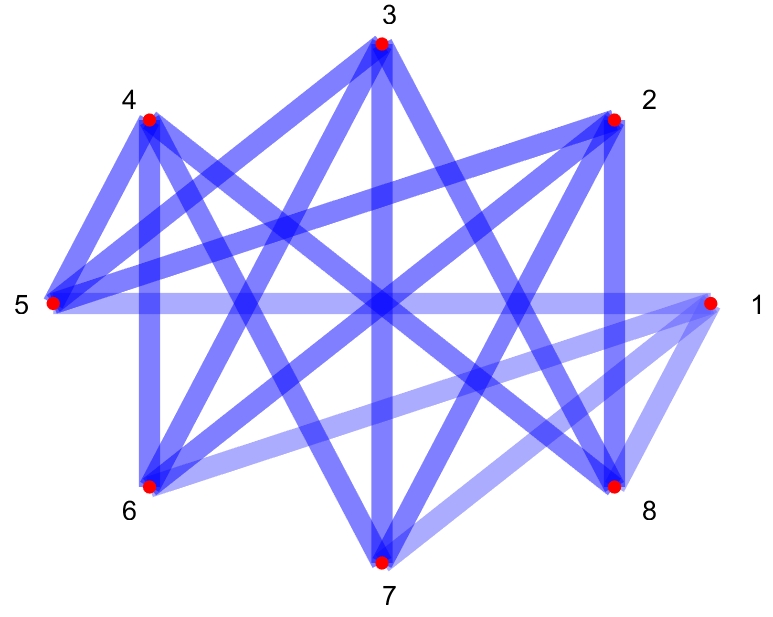}
\end{center}
\caption{\(C\geq 59\)} \label{fig:sim-59}
\end{subfigure}
\caption{Optimal Networks ($\phi=-0.05$)}
\end{figure}

When \(C\) is sufficiently large (e.g., Figure \ref{fig:sim-59}), the optimal network is isomorphic to the bi-partite network  \(\bar{w} K_{\lfloor \frac{n}{2} \rfloor, \lceil \frac{n}{2} \rceil}\). Moreover, based on the simulation results (Figures \ref{fig:sim-1} and \ref{fig:sim-8}), we observe that the weights of links \((1,5)\), \((1,4)\), and \((1,3)\) at \(C = 1\) are smaller than their corresponding weights at \(C = 0\) and \(C = 8\). Therefore, \(g_{15}^*, g_{14}^*, g_{13}^*\) are not monotonic in \(C\). More interestingly, the optimal network is isomorphic to \(\bar{w} K_{1,7}\) when \(C = 34\) (Figure \ref{fig:sim-34}), to \(K_{2,6}\) when \(C = 35\) (Figure \ref{fig:sim-35}), to \(\bar{w} K_{3,5}\) when \(C = 43\) (Figure \ref{fig:sim-43}), and finally to \(\bar{w} K_{4,4}\) when \(C = 59\) (Figure \ref{fig:sim-59}). We also observe that the optimal network tends to approach \(\bar{w} K_{2,6}\) when \(C\) is between 35 and 42.

Figure \ref{fig-sim-a} plots the absolute value of the cosine of the angle between \(\mathbf{a}^*(C)\) and the eigenvector corresponding to the largest eigenvalue of \([I - \phi \mathbf{g}^*(C)]^{-1}\). Figure \ref{fig:sim-lambda} depicts the evolution of \(\lambda_1(\phi \mathbf{g}^*)\) as a function of \(C\).

\begin{figure}[H]
\begin{subfigure}{0.45\textwidth}
\begin{center}
\includegraphics[scale=0.2]{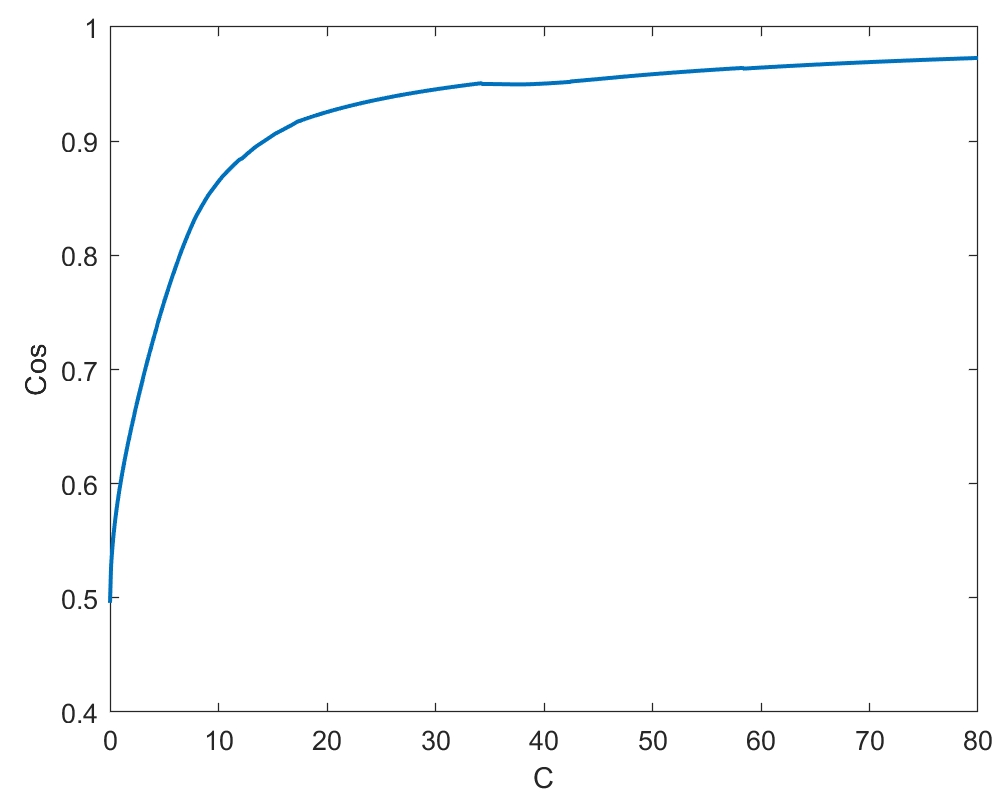}
\end{center}
\caption{The cosine of $\b a^*$ and the eigenvector.}
\label{fig-sim-a}
\end{subfigure}
\qquad
\begin{subfigure}{0.45\textwidth}
\begin{center}
\includegraphics[scale = 0.2]{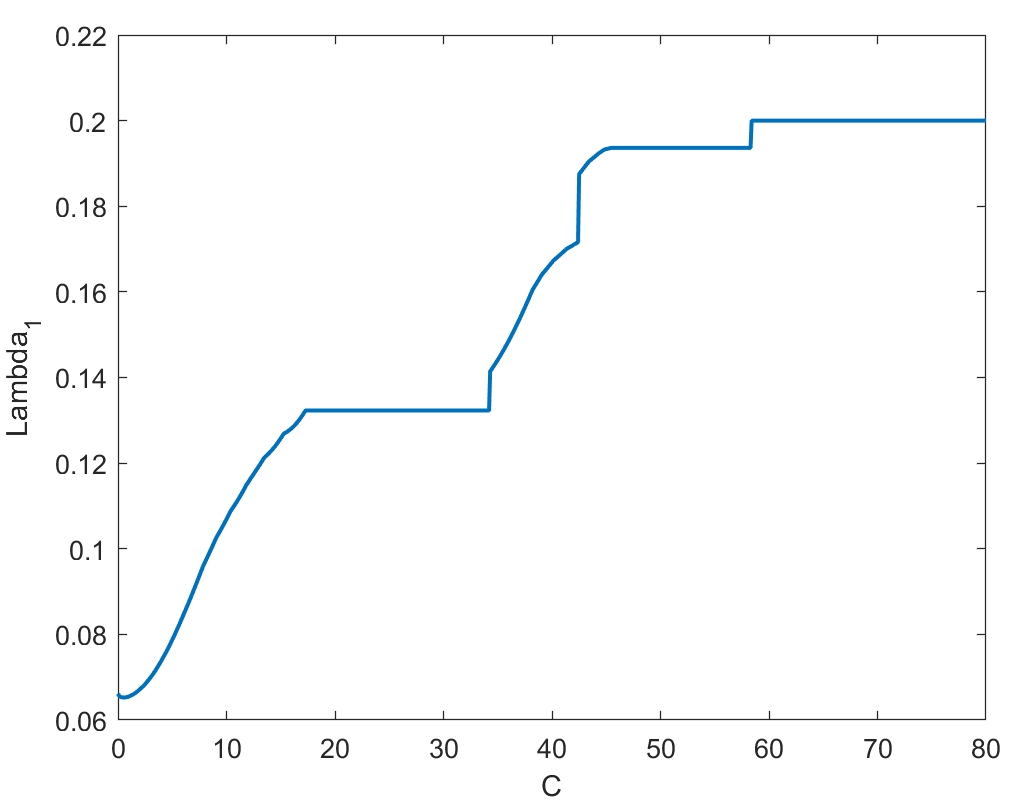}
\end{center}
\caption{\(\lambda_1(\phi\b g^*)\)).($\phi=0.05$)}\label{fig:sim-lambda}
\end{subfigure}
\end{figure}

There are two apparent stages when \(\phi < 0\). When the budget is relatively small, the optimal network is somewhat irregular and unpredictable. As shown in Figure \ref{fig:sim-lambda}, in order to balance the allocation of the budget between \(\b a^*\) and \(\b g^*\), \(\lambda_1(\phi \b g^*)\) may initially decrease, which facilitates intervention in \(\hat{\b a}\). When the budget becomes sufficiently large, the optimal network takes the form of a bipartite graph, corresponding to a local optimum. To reach the global optimum, one must compare these local optima; as \(C\) becomes large enough, \(\bar{w} K_{4,4}\), having the steepest slope, eventually becomes the globally optimal network.

There are three discontinuities in Figure \ref{fig:sim-lambda}: one occurs near \(C = 30\), another near \(C = 40\), and the third near \(C = 60\). Between \(C = 43\) and \(C = 60\), the remaining budget is entirely allocated to \(\mathbf{a}^*\) to further reduce the angle between \(\mathbf{a}^*\) and the eigenvector of \(\bar{w} K_{3,5}\). We observe that local optima arise at \(\mathbf{g} = \bar{w} K_{3,5}\) and \(\mathbf{g} = \bar{w} K_{4,4}\). Since it is less costly to intervene and adjust \(\hat{\mathbf{g}}\) to \(\bar{w} K_{3,5}\), the network \(\bar{w} K_{3,5}\) becomes globally optimal under moderate budgets. Consequently, the optimal solution exhibits discontinuities as a function of the budget \(C\). The joint intervention problem is complex due to its non-convex nature, featuring multiple local optima. Finally, Figures \ref{fig-sim+budget} and \ref{fig:sim-budget} compare the budget allocated to network interventions as a function of \(C\) for both \(\phi > 0\) and \(\phi < 0\).

\begin{figure}[H]
\begin{subfigure}{0.45\textwidth}
\begin{center}
\includegraphics[scale=0.2]{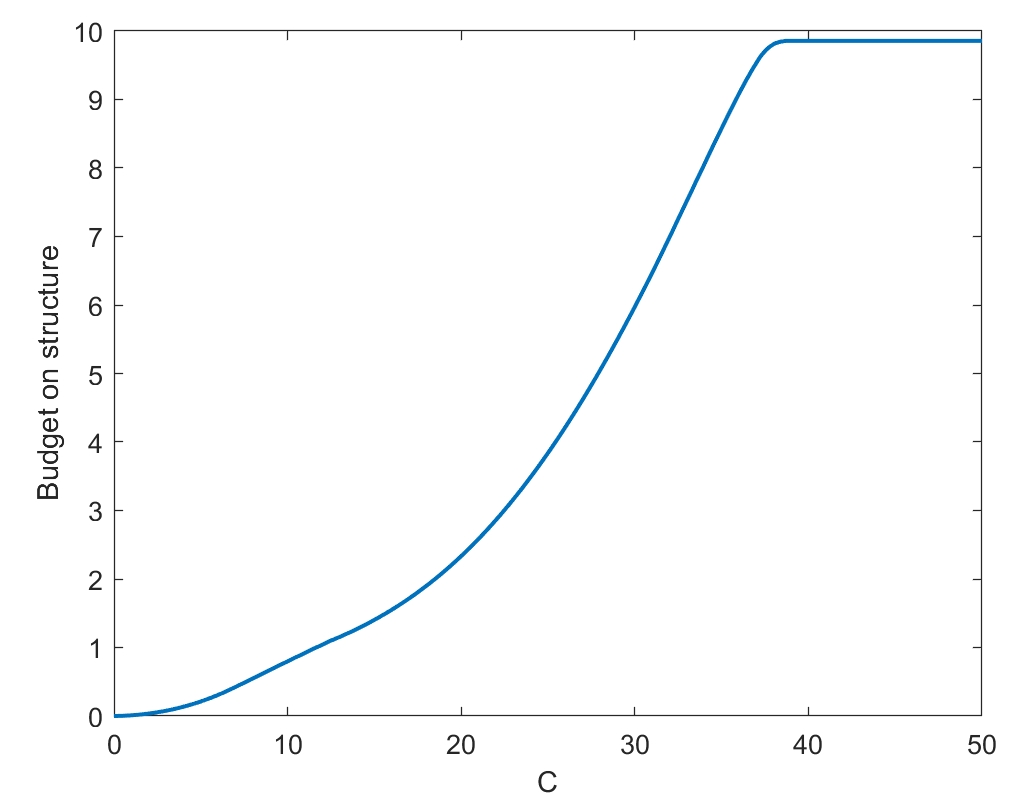}
\end{center}
\caption{Budget on the Network. ($\phi>0$)}
\label{fig-sim+budget}
\end{subfigure}
\qquad
\begin{subfigure}{0.45\textwidth}
\begin{center}
\includegraphics[scale = 0.2]{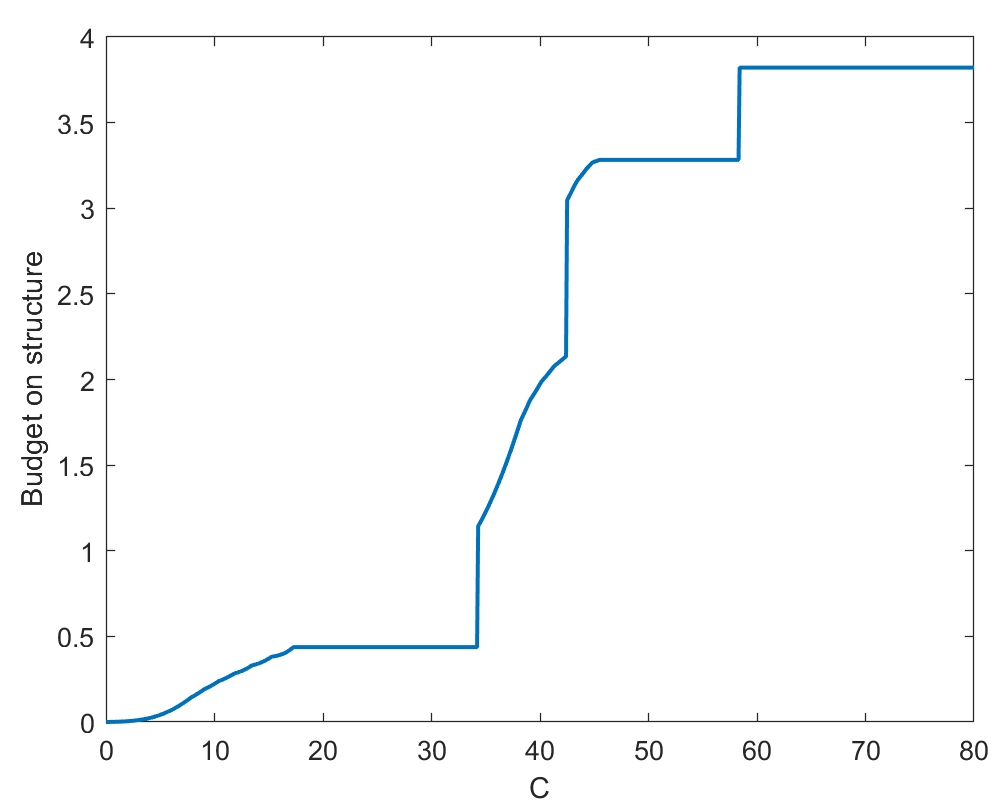}
\end{center}
\caption{Budget on the Network. ($\phi<0$)}\label{fig:sim-budget}
\end{subfigure}
\end{figure}

\section{Welfare and distributional effects}
\label{sec-4}
In this section, we analyze the effects of joint intervention on welfare and inequality. We compare the outcomes under joint and single interventions, with the single intervention serving as both a special case and an important benchmark for the analysis.

We first discuss welfare. To compare the optimal welfare under joint and single interventions, we consider the ratio
\[
r^*(\hat{\b g}, \hat{\b a}, C) = \frac{V_{\text{joint}}^*(\hat{\b g}, \hat{\b a}, C)}{V_{\text{single}}^*(\hat{\b g}, \hat{\b a}, C)}.
\]
When \(C\) is sufficiently large, the numerator of \(r^*\) is characterized by Theorem \ref{th-2}, while the denominator is characterized by \cite{ggg}. This allows us to provide an asymptotic characterization of \(r^*\) as follows.

\begin{theorem}\label{th-welfare}
Suppose Assumption \ref{as-1} holds.

(a) If \(\phi>0\), then \[\lim_{C\to\infty}r^*(\hb g,\hb a,C)=\left(\frac{1-\lambda_1(\phi\hb g)}{1-\lambda_1(\phi\bar w K_n)}\right)^2=\left(\frac{1-\phi\lambda_1(\hb g)}{1-(n-1)\phi\bar w)}\right)^2.\]

(b) If \(\phi<0\), then 
\[\lim_{C\to\infty}r^*(\hb g,\hb a,C)=\left(\frac{1-\lambda_1(\phi\hb g)}{1-\lambda_1(\phi\bar w K_{\lfloor\frac{n}{2}\rfloor,\lceil\frac{n}{2}\rceil})}\right)^2=\left(\frac{1-\phi\lambda_n(\hb g)}{1+\phi\bar w\sqrt{\left\lfloor\frac{n}{2}\right\rfloor\left\lceil\frac{n}{2}\right\rceil}}\right)^2.\]
\end{theorem}

By using the results from Theorem \ref{th-2} and Lemma \ref{lem-graph}, Theorem \ref{th-welfare} quantifies the value of the additional intervention in network design for large budgets by comparing the social multipliers under the initial network \(\hat{\b g}\) and the optimal network \(\b g^*\). Clearly, \(r^*(\hat{\b g}, \hat{\b a}, C) \geq 1\), since the planner's feasible set is larger under joint intervention than under single intervention. The expressions derived in Theorem \ref{th-welfare} also allow us to analyze how the welfare ratio is affected by the various model primitives.

\begin{corollary} \label{cor-2}Suppose Assumption \ref{as-1} holds. Then, \\
(a) \(\lim_{C\to\infty} r^*(\hb g,\hb a,C)\) is decreasing in \(\lambda_1(\phi\hb g)\).

(b) \(\lim_{C\to\infty} r^*(\hb g,\hb a,C)\) is increasing in \(|\phi|\).

(c) \(\lim_{C\to\infty} r^*(\hb g,\hb a,C)\) is independent of \(\kappa\) and \(\hb a\).
\end{corollary}

This corollary describes how model parameters shape the welfare ratio for a very large budget; it  follows directly from differentiating the expressions derived in Theorem \ref{th-welfare}. Corollary \ref{cor-2}(a) is intuitive, since the social multiplier under single intervention depends on the size of \(\lambda_1(\phi \hat{\b g})\). Corollary \ref{cor-2}(b) implies that joint intervention becomes more valuable as spillovers become stronger. Intuitively, the intensity of spillovers amplifies the effect of link modifications on players’ actions and, consequently, on total welfare. In contrast, Corollary \ref{cor-2}(c) suggests that the relative cost of link intervention plays a limited role in determining welfare, a point we further discuss in Section \ref{sec-5.1}.

What determines the distributional effects of interventions? Given two vectors \(\b v, \b w \in \mathbb{R}^n_+\) such that \(\b 1^T \b v = \b 1^T \b w = 1\), we say that \(\b v\) is more equitable than \(\b w\) (denoted \(\b v \succ_L \b w\)) if \(\b w\) majorizes \(\b v\).\footnote{That is, if we reorder the components such that \(v_{(1)} \geq \cdots \geq v_{(n)}\) and \(w_{(1)} \geq \cdots \geq w_{(n)}\), then \(\sum_{i=1}^k w_{(i)} \geq \sum_{i=1}^k v_{(i)}\) for all \(k \in \{1, \dots, n\}\).} 
Equivalently, interpreting \(\b v\) and \(\b w\) as wealth distributions, \(\b v\) Lorenz-dominates \(\b w\), which implies that \(\b v\) exhibits less inequality than \(\b w\) under standard measures such as the Gini coefficient and the Theil entropy index \citep{atkinson}. Moreover, the maximal elements under \(\succ_L\) are the vectors with equal entries, representing the case of perfect equality. 

Finally, we define the payoff distribution \(\mathcal{D}\) as the normalized vector of payoffs:
\[
\mathcal{D}^*_j(\hat{\b g}, \hat{\b a}, C) = \frac{\pi^*_j(\hat{\b g}, \hat{\b a}, C)}{\b 1^T \pi^*_j(\hat{\b g}, \hat{\b a}, C)}, \quad j \in \{\text{single}, \text{joint}\}.
\]
Since \(\b 1^T \mathcal{D}^*_j(\hat{\b g}, \hat{\b a}, C) = 1\), the relation \(\succ_L\) defines a partial order over \(\mathcal{D}^*_j(\hat{\b g}, \hat{\b a}, C)\), up to permutations of the indices.

\begin{theorem}\label{th-ineq}Suppose Assumption \ref{as-1} holds.

(a) The welfare-maximizing joint intervention achieves equality of payoffs as the budget goes to infinity when either \(\phi>0\) or \(\phi<0\) with \(n\) even.\footnote{When \(\phi<0\) and \(n\) is odd, the asymmetry in the optimal network for large budgets \(\b g^*=K_{\frac{n-1}{2},\frac{n+1}{2}}\) results in nonzero but low inequality.} That is, \(\lim_{C\to\infty}\frac{\pi^*_i}{\pi^*_j}=1\) for all \(i,j\).

(b) The welfare-maximizing joint intervention can induce a larger payoff inequality compared with single intervention.
That is, there exists a choice of parameters \(\hb g,\hb a,C\) such that \(\mathcal D^*_{single}(\hb g,\hb a,C)\succ_L\mathcal D^*_{joint}(\hb g,\hb a,C)\).
\end{theorem}

Theorem \ref{th-ineq} examines the distributional consequences of joint versus single interventions. Part (a) shows that when the budget is large enough, joint interventions can eliminate inequality entirely by producing uniform eigen-centralities. Part (b) shows that at moderate budget levels, network changes may increase inequality due to more unequal centrality distributions.

We first give an example to illustrate  Theorem \ref{th-ineq}(b).

\begin{example}\label{ex-4}
Let \(\kappa=0.5,\ \phi=0.15,\ \bar w=1,\ \hb a=\b 0\), and \[\hb g=\pmat{0&0.14&0.23&0.63&0.05\\0.14&0&0.25&0.14&0.46\\0.23&0.25&0&0.09&0.39\\0.63&0.14&0.09&0&0.11\\0.05&0.46&0.39&0.11&0}.\]
We have the following normalized payoff vectors:

\(\mathcal D^*_{single}(\hb g,\hb a,C)=\mathcal D^*_{joint}(\hb g,\hb a,0)=(0.217, 0.196, 0.188, 0.198, 0.2)^T\) for all \(C\).

\(\mathcal D^*_{joint}(\hb g,\hb a,4)=(0.199, 0.214, 0.201, 0.168, 0.219)^T\).

\(\mathcal D^*_{joint}(\hb g,\hb a,8)=(0.2, 0.2, 0.2, 0.2, 0.2)^T\).
\end{example}
We note that the inequality under single intervention is independent of \(C\), since the condition \(\hb a=\b 0\) implies that \(\b a^*\) is always an eigenvector of \(\b g^*\). Furthermore, it can be checked that 
\[
\mathcal D^*_{joint}(\hb g,\hb a,4)\succ_L\mathcal D^*_{joint}(\hb g,\hb a,0)=\mathcal D^*_{single}(\hb g,\hb a,C)\succ_L\mathcal D^*_{joint}(\hb g,\hb a,8),
\]
 illustrating that inequality increases under joint intervention when the planner's budget \(C=4\), which is where the graph intervention causes a dispersion in the  eigen-centralities of the optimal network, but inequality vanishes at the larger budget \(C=8\).\footnote{When \(C=4\), the optimal network is \(\b g^*=\pmat{0&0.71&0.80&1&0.62\\0.71&0&0.84&0.69&1\\0.80&0.84&0&0.64&0.99\\1&0.69&0.64&0&0.66\\0.62&1&0.99&0.66&0}\). When \(C=8\), the optimal network is the complete graph \(\b g^*=K_5\).}

To show Theorem \ref{th-ineq}(a), we make use of an important property of the equilibrium payoffs given in Lemma \ref{lem-2}.

\begin{lemma}\label{lem-2} Suppose the eigenspace corresponding to \(\lambda_1(\phi\b g^*)\) has dimension 1, and let \(\b u^1(\phi\b g^*)=(u_1^1,\cdots,u_n^1)\) be a representative unit eigenvector. Then
  \[\lim_{C\to\infty}\frac{\pi^*_i}{\pi^*_j}=\frac{(u^1_i)^2}{(u^1_j)^2}.\]
\end{lemma}
Intuitively, since each player's equilibrium payoff equals half of the square of their equilibrium effort, the relative payoff of two players is equal to the square of their relative Katz-Bonacich centralities. When \(C\to\infty\), \(\b a^*\)\textemdash hence, $\b x^*$\textemdash is approximately a principal eigenvector of \(\b g^*\) by equation \eqref{eq-foc1}. In other words, the relative standalone marginal utilities approximately equal the relative equilibrium efforts, which approximately equal the relative eigen-centralities---i.e.,
$
\frac{x_i^{*}}{x_j^{*}} \approx  \frac{a_i^*}{a_j^*} \approx \frac{u_i^1}{u_j^1} \text{ for all }i,j.
$
In combination, we obtain that 
  \begin{equation}\lim_{C\to\infty}\frac{\pi^*_i}{\pi^*_j}=\lim_{C\to\infty}\frac{x_i^{*2}}{x_j^{*2}}=\lim_{C\to\infty}\frac{a^{*2}_i}{a^{*2}_j}=\frac{(u^1_i)^2}{(u^1_j)^2}.\label{eq-ratio}\end{equation}
  
Consequently, the payoff inequality in the limit is solely determined by the inequality of the squared entries of the principal eigenvector \(\b u^1(\phi\b g^*)\).

Under joint intervention, when \(C\) is large, the planner selects either the complete network or the complete bipartite network, as characterized in Theorem \ref{th-2}. The principal eigenvectors of these networks are provided in Fact \ref{fact-1}. We observe that, by allowing for endogenous network formation, payoff inequality can be entirely eliminated for large \(C\) when either \(\phi > 0\) or \(\phi < 0\) and \(n\) is even, since in these cases we have \(|u^1_i(\phi \b g^*)| = \frac{1}{\sqrt{n}}\) for all \(i\).

\begin{remark} \label{rmk-5}
For large budgets, asymptotically zero inequality is achieved under single intervention if and only if \(|u^1_i(\phi \hat{\b g})| = \frac{1}{\sqrt{n}}\) for all \(i\). When \(\phi > 0\), this occurs if and only if \(\hat{\b g}\) is regular.
\end{remark}

Remark \ref{rmk-5} clarifies when the single-intervention planner can also eliminate inequality: only if the initial network is regular (equal degrees for all agents). Otherwise, even large budgets cannot achieve equal payoffs without network redesign.

\begin{remark}\label{rmk-LB}
For large budgets, since \(\pi^*_i / \pi^*_j \approx a^{*2}_i / a^{*2}_j\), joint intervention also results in approximately equal intervention levels in the standalone marginal utilities \(\b a\) across agents, whereas substantial heterogeneity in intervention may arise under single intervention depending on the principal components of \(\hat{\b g}\).
\end{remark}

This remark links the earlier lemma with intervention design. Equal payoffs imply equal efforts and hence equal standalone utilities. In contrast, if the network is fixed and irregular, effort and utility interventions must be unequal, limiting equality.

From Theorems \ref{th-welfare} and \ref{th-ineq}, we conclude that, for large budgets, allowing for joint intervention improves both total welfare and payoff equality. Therefore, a planner should implement both targeted interventions and network design in order to simultaneously achieve the dual social objectives of maximizing welfare and minimizing inequality.

However, for intermediate budgets, network changes may actually increase inequality, generating a trade-off between welfare and inequality. Focusing on the case \(\phi > 0\), such an increase in inequality is particularly pronounced when \(\hat{\b g}\) is close to being regular but not vertex-transitive.\footnote{A graph \(\b g\) is vertex-transitive if, for any two vertices \(\b v\) and \(\b v'\), there exists an automorphism \(\Psi\) on \(\b g\) such that \(\Psi(\b v) = \b v'\).} In this case, inequality at \(\hat{\b g}\) is initially low, but as the budget increases, there exists a range of \(C\) in which the welfare-maximizing network \(\b g^*\) deviates from regularity as it progressively approaches the complete graph \(K_n\), as illustrated in Example \ref{ex-4}.

\subsection{The welfare cost of equality}

Previously, we identified the potential adverse effect of joint intervention on payoff inequality for a welfare-maximizing planner. We now turn to studying the welfare cost of imposing zero payoff inequality under both single and joint interventions. To analyze this trade-off, we consider a related problem in which the planner prioritizes minimizing inequality over maximizing total welfare. Specifically, the planner solves \eqref{eq-prob} subject to the additional constraint \(\pi_i^* = \pi_j^*\) for all \(i, j\).

We show that the welfare loss from this equity constraint is negligible under joint intervention when \(n\) is large, but it can be substantial under single intervention. By Theorem \ref{th-ineq}(a), for large \(C\), if \(\phi > 0\), or \(\phi < 0\) and \(n\) is even, then \(\frac{\pi_i^*}{\pi_j^*} \to 1\), implying asymptotically zero inequality. We analyze the remaining case, where \(\phi < 0\) and \(n\) is odd, in the following lemma.

\begin{lemma}\label{lem-3}
Let \(n\geq 5\) be odd, and \(\b g\in\mathcal G_n\) such that \(|u^n_i(\b g)|=|u^n_j(\b g)|\) for all \(i,j\). Then \[\lambda_n(\b g)\geq-\bar w\left(\frac{n-1}{2}\right),\]
with equality when \[g_{ij}=\begin{cases}0,&i,j\leq \frac{n+1}{2};\\1,&i\leq \frac{n+1}{2}<j\text{ or }j\leq \frac{n+1}{2}<i;\\\frac{2}{k-3},& i,j>\frac{n+1}{2}\text{ and }i\neq j.\end{cases}\]
\end{lemma}

Lemma \ref{lem-3} is analogous to Lemma \ref{lem-graph}, but restricted to network structures whose smallest eigenvector has entries of equal magnitude (in absolute value). Such choices of \(\b g^*\) yield asymptotically zero inequality for large \(C\), due to the proportionalities established in \eqref{eq-ratio}. From Theorem \ref{th-2}, we know that, when \(\phi < 0\) and \(C\) is large, optimal welfare depends critically on the lower bound of \(\lambda_n(\b g)\). Therefore, comparing the lower bounds for \(\lambda_n(\b g)\) provided in Lemmas \ref{lem-graph} and \ref{lem-3} allows us to quantify the welfare cost of imposing equality.

Moreover, we observe that the ratio of these lower bounds satisfies
\[
\lim_{n \to \infty} \frac{-\bar{w} \left( \frac{n-1}{2} \right)}{-\bar{w} \sqrt{\frac{n^2 - 1}{2}}} = 1,
\]
implying that, as \(n\) grows large, the welfare loss from imposing payoff equality becomes negligible.

By contrast, under single intervention, inequality is closely tied to the structure of the initial network \(\hat{\b g}\). As a result, the planner may incur a much larger welfare loss in order to achieve equality.

\begin{proposition}\label{prop-4}
Suppose Assumption \(\ref{as-1}\) holds, \(\phi>0\) and \(\hb a=\b 0\).\footnote{For general \(\hb a\), similar results hold for \(C\to\infty\).} Then for any choice of single intervention such that \(\pi_i^*=\pi_j^*\) for all \(i,j\), the total welfare satisfies \[V_{single,eq}(\hb g,\hb a,C)\leq\frac{1}{\|[\b I-\phi\hb g]\b z\|^2},\] where \(\b z=\frac{1}{\sqrt n}\b 1_n\) is the normalized vector of ones.
\end{proposition}

Proposition \ref{prop-4} quantifies the welfare loss from enforcing equal payoffs under single intervention. It shows that if the network is not regular, enforcing equality forces the planner to use a suboptimal direction in the space of utility interventions, which reduces efficiency. This loss persists even with large budgets, unlike in the joint intervention case.

Specifically,  if \(\hat{\b g}\) is not regular, then \(\b z\) is not an eigenvector of \(\hat{\b g}\), and thus 
\begin{equation*}
 V_{\text{single,eq}}(\hat{\b g}, \hat{\b a}, C) / V_{\text{single}}(\hat{\b g}, \hat{\b a}, C)
\end{equation*}
 is strictly less than 1. In other words, under single intervention, an inequality-minimizing planner achieves a lower total payoff than a utilitarian planner whenever \(\hat{\b g}\) is not regular. 
Since the optimal joint intervention is identical for both types of planners, the gains in total payoff for the inequality-minimizing planner are even larger than those obtained by a utilitarian planner. As a result, our bounds in Theorem \ref{th-welfare} continue to apply under this lexicographic social welfare function.

In summary, we find that allowing for joint intervention not only improves welfare and reduces inequality, but also weakens the trade-off between inequality and total payoff, leading to even greater welfare improvements when the planner explicitly accounts for social inequality.

\subsection{Simulations of the welfare and equality for intermediate budgets}

Here, we provide simulation results of Example \ref{ex-2}. Figure \ref{fig:welfare-1} displays  the welfare under the optimal single intervention (red line) and the optimal joint intervention (blue line) while  Figure \ref{fig:welfare-r} plots $r^*$ with respect to the budget $C$.

\begin{figure}[H]
\begin{subfigure}{0.45\textwidth}
\begin{center}
\includegraphics[scale=0.2]{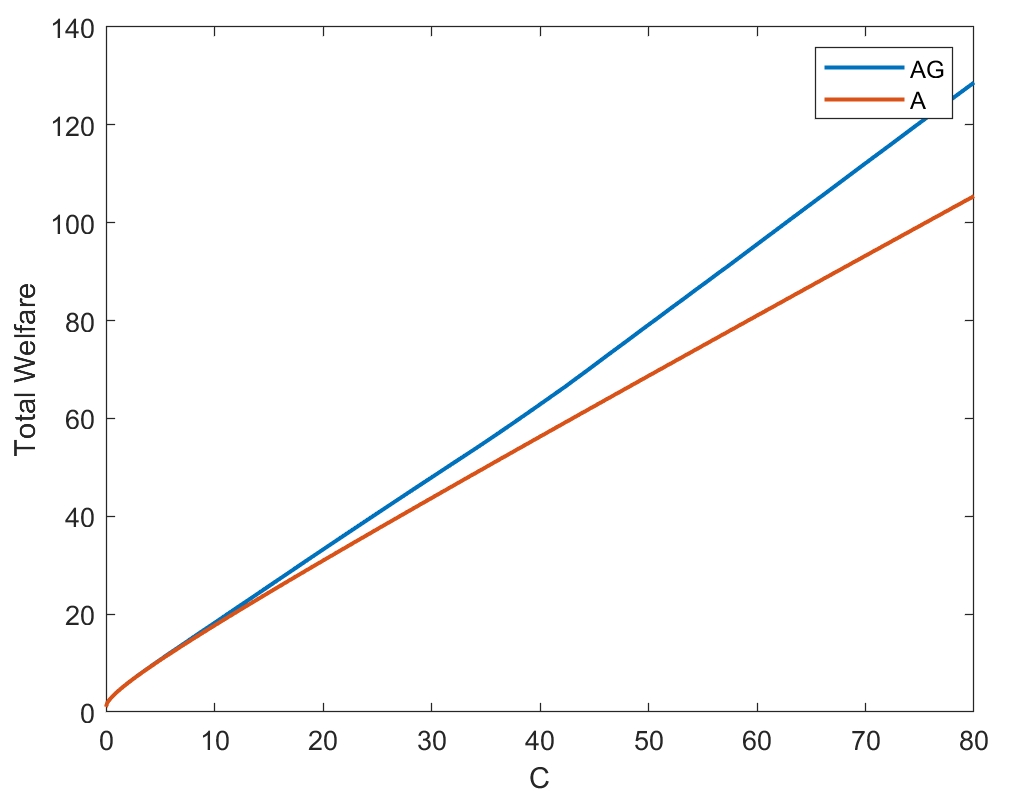}
\end{center}
\caption{Welfare (Example \ref{ex-2})}\label{fig:welfare-1}
\end{subfigure}
\qquad
\begin{subfigure}{0.45\textwidth}
\begin{center}
\includegraphics[scale = 0.2]{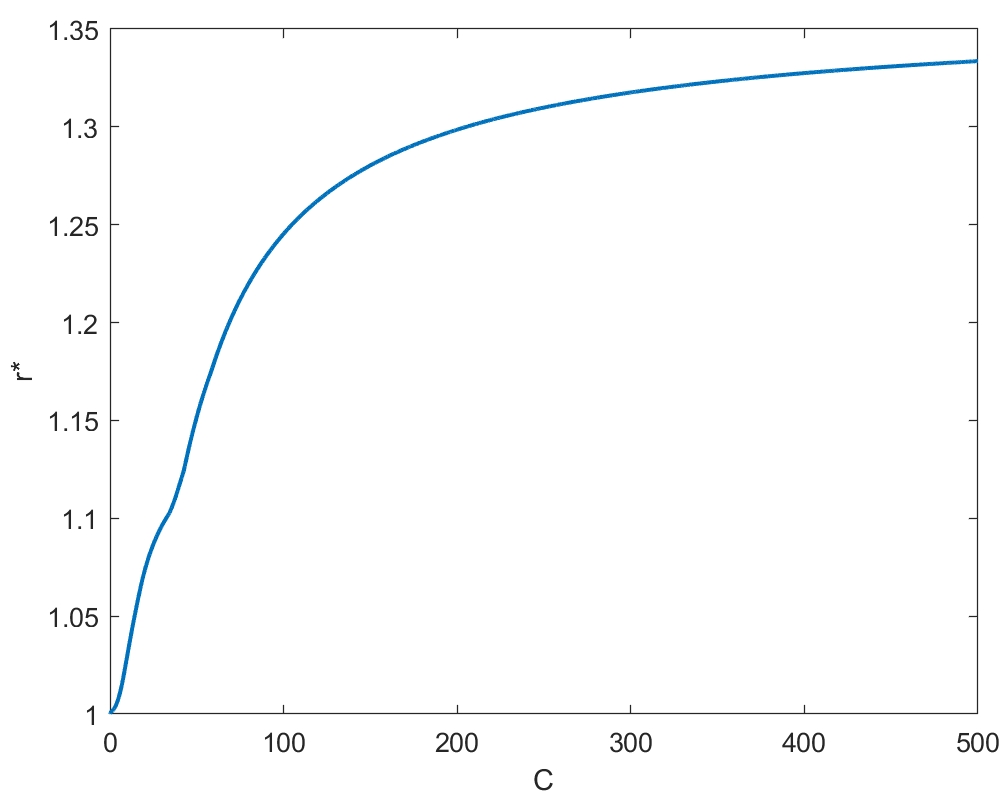}
\end{center}
\caption{\(r^*\) (Example \ref{ex-2})}\label{fig:welfare-r}
\end{subfigure}
\caption{Welfare simulations ($\phi<0$)} \label{fig:welfare}
\end{figure}

To evaluate the equality index of the payoffs of $n$ players, we calculate the Gini index of the payoffs under optimal joint and single intervention. In Figure \ref{fig:sim-gini-}, the red line represents the Gini index under single intervention and the blue line represents the Gini index under joint intervention with respect to the budget $C$.

\begin{figure}[H]
\centering
\includegraphics[scale = 0.2]{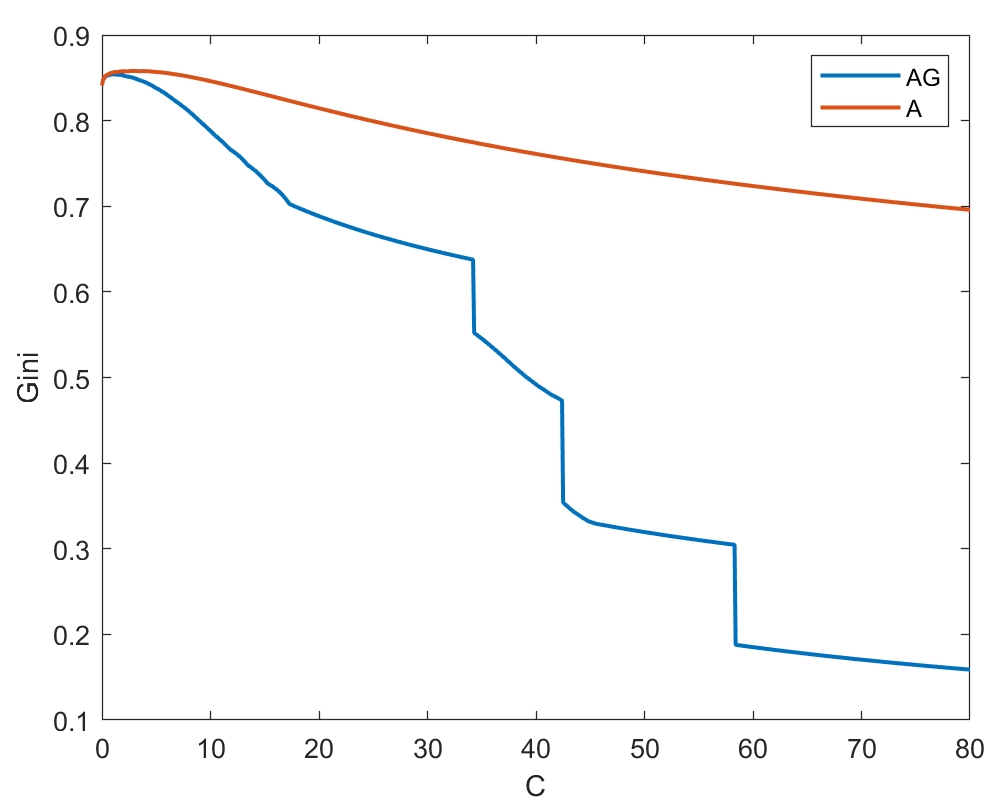}

\caption{Gini Index ($\phi<0$)}\label{fig:sim-gini-}
\end{figure}

Similarly, when $\phi=0.05$ and other settings are the same as example \ref{ex-2}, we simulate the Gini index in Figure \ref{fig:sim-gini+} and the $r^*$ in Figure \ref{fig:sim-r+} .

\begin{figure}[H]
\begin{subfigure}{0.45\textwidth}
\begin{center}
\includegraphics[scale=0.2]{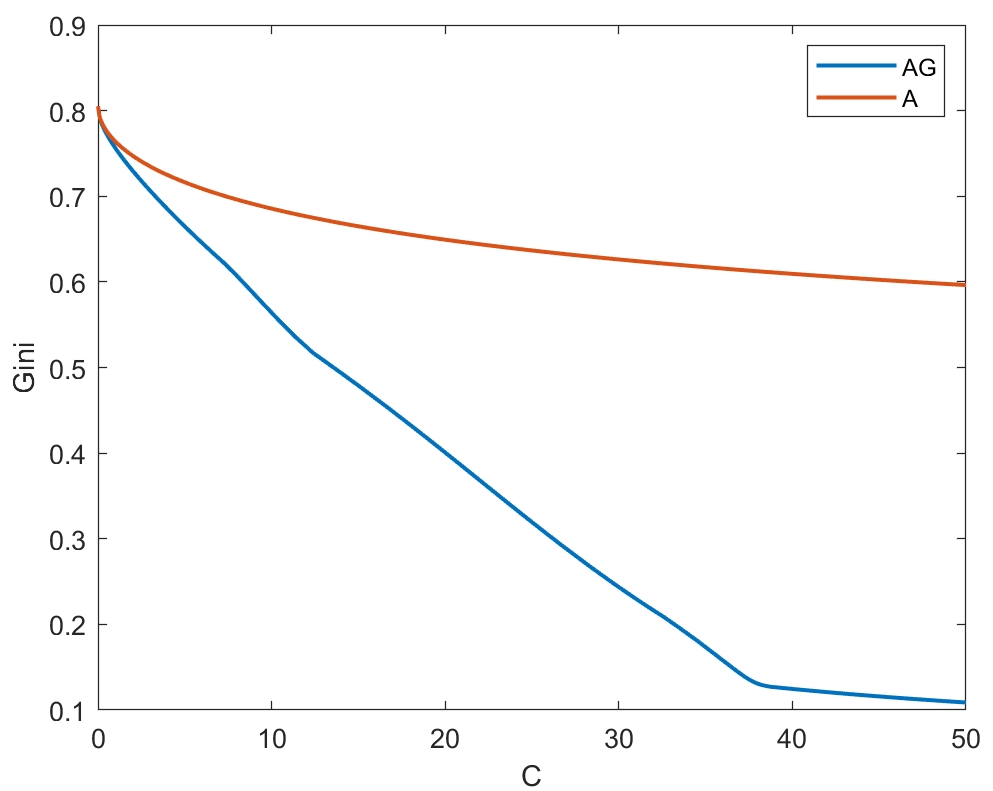}
\end{center}
\caption{Gini Index}\label{fig:sim-gini+}
\end{subfigure}
\qquad
\begin{subfigure}{0.45\textwidth}
\begin{center}
\includegraphics[scale = 0.2]{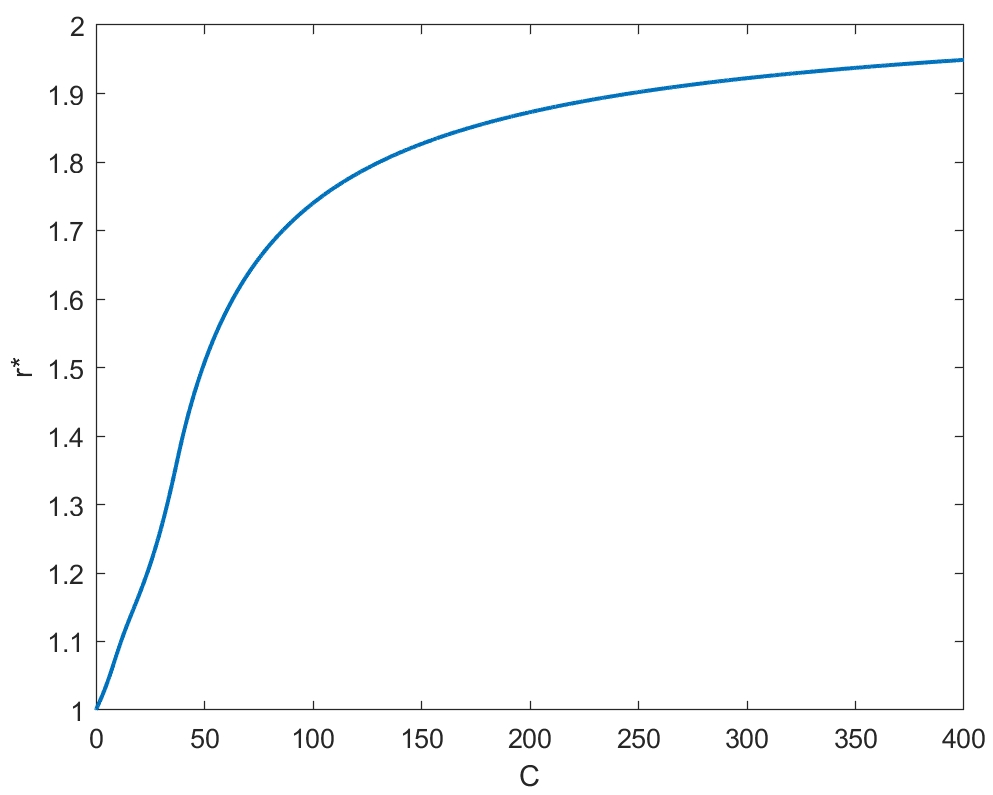}
\end{center}
\caption{\(r^*\)}\label{fig:sim-r+}
\end{subfigure}
\caption{Welfare simulations ($\phi>0$)}\label{fig:sim+}
\end{figure}

From the Figures \ref{fig:welfare} and \ref{fig:sim+}, as predicted, \(r^*\) converges to the ratio of the largest eigenvalues of \((I - \phi \hat{\b g})^2\) and \((I - \phi \bar{\b g})^2\), where \(\bar{\b g} = \bar{w} K_8\) for \(\phi > 0\) and \(\bar{\b g} = \bar{w} K_{4,4}\) for \(\phi < 0\), as \(C\) becomes large. In addition, the Gini index under joint intervention converges to zero. 
However, when \(C\) corresponds to an intermediate budget, the effect of \(C\) on inequality becomes ambiguous. For example, as shown in Figure \ref{fig:sim-gini-}, the Gini index is not monotonic.

\section{Extensions}\label{sec-5}

\subsection{Pure network design}\label{sec-5.2}

As a counterpart to the single intervention analyzed in \cite{ggg}, we now consider the case in which the planner can intervene only in the network structure \(\b g\), while the standalone utilities   remain fixed at their pre-intervention values  \(\b a\). This corresponds to the optimization problem:
\begin{equation}
\max_{\b g\in\mathcal G_n}\quad \hb a^T[\b I-\phi\b g]^{-2}\hb a.
\label{eq-pure}
\end{equation}
Following the earlier discussion on general cost functions, and noting that the cost of modifying the network is bounded, we omit the budget constraint without loss of generality, as it does not bind in this case. Consequently, the total achievable welfare is also bounded. As a result, pure network design yields lower welfare gains than characteristic-based (utility) interventions when the planner has access to a sufficiently large budget.
Interestingly, although the optimal network configuration under pure design differs from those obtained under joint interventions, its structure retains similar features. In particular, the optimal network generally consists of connected components that are either complete or complete bipartite graphs.

\begin{proposition}\label{prop-pure}
Suppose \(\hat a_i\neq0\) for all \(i\), and Assumption \ref{as-1} holds. Define the sets \(A^+=\{i:\hat a_i>0\}\) and \(A^-=\{i:\hat a_i<0\}\). Then the solution to \eqref{eq-pure} satisfies the following:

(a) If \(\phi\to 0^+\), then \(\b g^*\) consists of two disjoint complete graphs formed by the sets \(A^+\) and \(A^-\).

(b) If \(\phi\to0^-\), then \(\b g^*\) is the complete bipartite graph with the partitions as the sets \(A^+\) and \(A^-\).
\end{proposition}
Proposition \ref{prop-pure} follows from a Taylor expansion of the objective function: 
\[\hb a^T[\b I-\phi\b g]^{-2}\hb a=\hb a^T\hb a+2\phi\sum_{i\neq j}g_{ij}\hat a_i\hat a_j+\mathcal{O}(\phi^2).\]
Since \(\hb a^T\hb a\) is exogenous, if \(|\phi|\) is sufficiently small, then maximizing \(\hb a^T[\b I-\phi\b g]^{-2}\hb a\) will be equivalent to maximizing the linear term \(2\phi\sum_{i\neq j}g_{ij}\hat a_i\hat a_j\). Since \(g_{ij}\in[0,\bar w]\), the maximum is achieved under the condition that \(g_{ij}=\bar w\) if \(\phi\hat a_i\hat a_j>0\), and \(g_{ij}=0\) otherwise. Therefore, when \(\phi>0\), then links are formed between \(i\) and \(j\) when \(\hat a_i,\hat a_j\) are of the same sign, while when \(\phi<0\), then links are formed between \(i\) and \(j\) when \(\hat a_i,\hat a_j\) are of opposite signs. This gives the characterization in Proposition \ref{prop-pure}.

\subsection{Alternative objective functions and cost functions}\label{sec-5.1}

Here, we show that although we made use of quadratic utilities and costs in the previous sections, our results for large budgets are robust to a variety of functional forms. Consider the case where the planner has an objective function \(f(\pi_1,\cdots,\pi_n)\), so the planner solves 
\begin{align*}
\max_{\b a\in\mathbb R^n,\ \b g\in\mathcal G_n}\quad& f(\pi_1,\cdots,\pi_n)\\\text{s.t.}\quad&\kappa\|\b g-\hb g\|^2+\|\b a-\hb a\|^2\leq C.\stepcounter{equation}\tag{\theequation}\label{eq-util}
\end{align*}

In \eqref{eq-util}, \(f\) represents the choice of social welfare function implemented by the planner. In our base model, we have considered the case of a utilitarian planner, with \(f(\pi_1,\cdots,\pi_n)=\sum_{i=1}^n\pi_i\). A possible alternative is the Rawlsian utility function \(f(\pi_1,\cdots,\pi_n)=\min_i\pi_i\), where the planner aims to maximize the lowest utility obtained across all players. Both of these cases are covered under the following proposition.

\begin{proposition}\label{prop-util}
Suppose \(f\) is symmetric, increasing, concave, and Assumption \ref{as-1} holds. Then as \(C\) goes to infinity, the solution to \eqref{eq-util} tends to \(\b g^*=\bar wK_n\) if \(\phi>0\), and tends to \(\b g^*=\bar wK_{\frac{n}{2},\frac{n}{2}}\) if \(\phi<0\) and \(n\) is even.
\end{proposition}

To show Proposition \ref{prop-util}, we note that when \(\b g=\bar wK_n\) or \(\b g=\bar wK_{\frac{n}{2},\frac{n}{2}}\) depending on the sign of \(\phi\), the payoffs for each agent are asymptotically equal by Theorem \ref{th-ineq}. The sum of payoffs \(\sum_{i=1}^n\pi_i\) is also maximized by Theorem \ref{th-2}. Since \(f\) is symmetric and concave, we have \(f(\pi_1,\cdots,\pi_n)\leq f(\bar\pi,\cdots,\bar\pi)\) for any \(\pi_1,\cdots,\pi_n\), where \(\bar\pi=\frac{1}{n}\sum_{i=1}^n\pi_i\). Furthermore, \(f(\bar\pi,\cdots,\bar\pi)\) is increasing in \(\bar\pi\). Thus the optimal payoff is attained when \(\b g=\bar wK_n\) or \(\b g=\bar wK_{\frac{n}{2},\frac{n}{2}}\).

We next allow for more general cost functions instead, and solve the optimization problem
\begin{align*}
\max_{\b a\in\mathbb R^n,\ \b g\in\mathcal G_n}\quad& \b a^T[\b I-\phi\b g]^{-2}\b a\\\text{s.t.}\quad&h(\b g;\hb g)+\|\b a-\hb a\|^2\leq C.\stepcounter{equation}\tag{\theequation}\label{eq-cost}
\end{align*}

Here, the function \(h\) represents the cost of intervention in the network structure.\footnote{The problem of general cost functions for intervention on the \(\b a\) component is studied in \cite{ggg}.} In our base model, we have assumed that \(h\) is given by the square of the \(L_2\)-norm, \(h(\b g;\hb g)=\sum_{i\neq j}(g_{ij}-\hat g_{ij})^2\), which helped to simplify the characterization in Proposition \ref{th-1}. However, many other cost functions are possible, such as the \(L_1\)-norm, \(h(\b g-\hb g)=\sum_{i\neq j}|g_{ij}-\hat g_{ij}|\). The appropriate choice of cost function will depend on the policies and technologies available to the planner, but we show in the following proposition that the optimal network structure is independent of the cost function for large budgets. 

\begin{proposition}\label{prop-cost}
Suppose \(h\) is continuous, and Assumption \ref{as-1} holds. As \(C\) goes to infinity, the solution to \eqref{eq-cost} tends to \(\b g^*=\bar wK_n\) if \(\phi>0\), and tends to \(\b g^*=\bar wK_{\frac{n}{2},\frac{n}{2}}\) if \(\phi<0\) and \(n\) is even.
\end{proposition}
To show Proposition \ref{prop-cost}, first observe that since \(\mathcal G_n\) is compact, then the expenditure on network design \(h\) is bounded. By an argument analogous to \eqref{eq-o}, we find that as \(C\) grows, the cost on network design becomes irrelevant and the dominant term for the total welfare is will still be the social multiplier \(\frac{1}{(1-\lambda_1(\phi\b g))^2}\), so the graphs that maximize \(\lambda_1(\phi\b g)\) (see Lemma \ref{lem-graph}) will be optimal for large \(C\).

\section{Concluding remarks}\label{sec_Con}

In many economic and social environments, a planner can influence both individuals' incentives and the network through which their actions interact. Such \emph{joint interventions}, where the planner simultaneously modifies individuals' private returns to investment and the structure of the network, are increasingly relevant in applications where the network is exogenous to individual users, such as infrastructure networks. Despite their growing importance, most of the existing literature has focused on targeted interventions along a single margin, either by modifying individual incentives or by altering the network structure. This paper develops a general framework to analyze the design of optimal joint interventions and highlights their implications for welfare and inequality.

We provide a tractable characterization of the optimal intervention problem under quadratic costs and strategic interactions, showing how the planner simultaneously allocates the budget across private returns and network weights. Our theoretical results establish that the optimal network adopts simple structures in large budgets: either a complete network under strategic complements, or a complete balanced bipartite network under strategic substitutes. These results allow us to quantify both the welfare gains and the inequality implications of joint interventions relative to single interventions.

While joint interventions always yield higher welfare by expanding the planner's feasible set, we show that they are particularly effective in simultaneously improving welfare and reducing payoff inequality, especially for large budgets. However, we also document that for intermediate budgets, network adjustments may introduce nontrivial trade-offs between welfare and inequality, depending on the initial network structure. Our results highlight that incorporating network design into intervention policies can substantially reduce these trade-offs and enhance policy effectiveness.

Overall, our analysis demonstrates that jointly targeting individuals' incentives and the network structure can lead to significant improvements in both efficiency and equity, providing novel insights for the design of optimal interventions in networked environments.

One interesting direction for research is when the network forms endogenously through the choices of individual players, rather than being directly designed by the planner. In such models (e.g. the framework studied by \cite{sg}), each player decides which links to form, and the network structure arises as an equilibrium outcome of their collective decisions.  A social planner in this context cannot choose the network outright but can intervene indirectly by influencing the incentives for link formation. For example, the planner might subsidize the creation of certain beneficial links or impose taxes/fees on forming certain links to discourage them. Analyzing the optimal subsidy or tax scheme for link formation in an endogenous network game is a promising avenue for future work, as it would offer a new perspective on targeted interventions that align individual incentives with social welfare objectives.

Additionally, when interactions are strategic substitutes ($\phi <0$), there is a computational challenge in implementing the optimal network design. Choosing the best way to partition the players into two groups (to form the optimal balanced bipartite network) is an NP-hard problem (Proposition \ref{prop-np}).  This means that there is no known efficient algorithm to find the optimal bipartition for large networks, making exhaustive search infeasible as the network size grows. Further work is thus needed to develop approximation algorithms or heuristic methods that can guide the planner’s decisions in this scenario. Designing such algorithms would reduce the computational difficulty and enable near-optimal network interventions even when the exact optimum is too complex to compute.

\newpage
\appendix

\strut

\begin{center}{\bf \Large Appendix: Proofs of the results in the text}\end{center}

\bigskip 

\begin{proof}{Theorem \ref{prop-1}}
We begin by showing that the budge constraint must be binding under the optimal solution. If the budget constraint is not binding for the solution $(\b a^*,\b g^*)$, then there must be a parameter $\lambda>1$ such that $(\lambda\b a^*,\b g^*)$ satisfies the budget constraint. Since $\lambda^2 V(\b a^*,\b g^*;\hb g,\hb a, C)= V( \lambda\b a^*,\b g^*;\hb g,\hb a, C)$, $(\b a^*,\b g^*)$ cannot be optimal. ($V(\b a^*,\b g^*;\hb g,\hb a, C)>0$ since $(\b I-\phi \b g^*)^{-2}$ must be positive definite.)

Suppose $(\b a^*,\b g^*)$ is optimal. Let $L(\b a^*,\b g^*)$ be the Lagrangian of the Problem \ref{eq-prob}. Therefore, $L(\b a^*,\b g^*)=\b a^{*T}[\b I-\phi\b g^*]^{-2}\b a^*+\mu (C-\kappa\|\b g^*-\hb g\|^2-\|\b a^*-\hb a\|^2)$ where $\mu=\frac{\partial V^*}{\partial C}$. 

(\ref{eq-foc1}) is just the FOC of $L$ with respect to \(\b a^*\) (recall that $\b g^*$ is symmetric): \begin{equation} 2 [\b I-\phi\b g^*]^{-2}\b a^*=2\mu(\b a^*-\hb a).\label{eq-foca}\end{equation}
Rewriting (\ref{eq-foca}) with respect to the basis $\{\b u_1,\cdots,\b u_n\}$ gives (\ref{eq-foc1}).

For (\ref{eq-foc2}), we first observe that $L(\b a^*,\b g^*) \geq  L(\b a^*, (1-t) \b g^* +t \b g')$ for any $t\in[0,1]$ and $\b g'\in\mathcal G_n$. Thus the directional directive of $L(\b a^*, \cdot)$, in the direction of $\b g'-\b g^*$ must be nonpositive. We evaluate them in the following Lemma:

\begin{lemma}{(Some matrix calculus results)}\label{lem-dir}
Define \[\c H=\{\b h\in\mathbb R^{n\times n}|h_{ij}=h_{ji}\text{ and }h_{ii}=0\text{ for all }i,j.\}.\]
(a) As a function of the network \(\b g\), the directional derivative of \(\b a^T[\b I-\phi\b g]^{-2}\b a\) in the direction of \(\b h\in\c H\) equals \begin{equation}\hspace{-1cm}\lim_{\epsilon\to0}\frac{\b a^T[\b I-\phi(\b g+\epsilon\b h)]^{-2}\b a-\b a^T[\b I-\phi\b g]^{-2}\b a}{\epsilon}=2Tr(\phi[\b I-\phi\b g]^{-1}\b a\b a^T[\b I-\phi\b g]^{-2}\b h).\label{eq-lem1a}\end{equation}

(b) As a function of the network \(\b g\), the directional derivative of \(\|\b g-\hb g\|^2\) in the direction of \(\b h\in\c H\) equals \begin{equation}
    \lim_{\epsilon\to0}\frac{\|\b g+\epsilon\b h-\hb g\|^2-\|\b g-\hb g\|^2}{\epsilon}=2Tr((\b g-\hb g)\b h).\label{eq-lem1b}
\end{equation}
\end{lemma}
\begin{proof}{Lemma \ref{lem-dir}} 
The proof follows from straightforward matrix operations.

(a) \begin{align*}&\hspace{-1cm}\lim_{\epsilon\to0}\frac{\b a^T[\b I-\phi(\b g+\epsilon\b h)]^{-2}\b a-\b a^T[\b I-\phi\b g]^{-2}\b a}{\epsilon} =\lim_{\epsilon\to0}\frac{\|[\b I-\phi(\b g+\epsilon\b h)]^{-1}\b a\|^2-\|[\b I-\phi\b g]^{-1}\b a\|^2}{\epsilon}\\=&\lim_{\epsilon\to0}\frac{\langle([\b I-\phi(\b g+\epsilon\b h)]^{-1}+[\b I-\phi\b g]^{-1})\b a,([\b I-\phi(\b g+\epsilon\b h)]^{-1}-[\b I-\phi\b g]^{-1})\b a\rangle}{\epsilon}\\=&\langle2[\b I-\phi\b g]^{-1}\b a,\lim_{\epsilon\to0}\frac{([\b I-\phi(\b g+\epsilon\b h)]^{-1}-[\b I-\phi\b g]^{-1})}{\epsilon}\b a\rangle\\
=&2\langle[\b I-\phi\b g]^{-1}\b a,[\b I-\phi\b g]^{-1}\phi\b h[\b I-\phi\b g]^{-1}\b a\rangle=2Tr(\phi[\b I-\phi\b g]^{-1}\b a\b a^T[\b I-\phi\b g]^{-2}\b h).\end{align*}

(b)
\begin{align*}\lim_{\epsilon\to0}\frac{\|\b g+\epsilon\b h-\hb g\|^2-\|\b g-\hb g\|^2}{\epsilon}=\lim_{\epsilon\to0}\frac{\langle2\b g+\epsilon\b h-2\hb g,\epsilon\b h\rangle}{\epsilon}&=\langle2(\b g-\hb g),\b h\rangle
 =2Tr((\b g-\hb g)\b h).\end{align*}
\end{proof}

Applying equations \eqref{eq-lem1a} and \eqref{eq-lem1b} in Lemma \ref{lem-dir}, we obtain that for any \(\b g'\in\mathcal G_n\),  \[\langle \left\{\phi[\b I-\phi\b g^*]^{-1}\b a^*\b a^{*T}[\b I-\phi\b g^*]^{-2}-\mu^*\kappa(\b g^*-\hb g) \right\},\b g'-\b g^*\rangle\leq0.\]

Define \(\b e_{ij}\) to be a matrix with 1 on the \((i,j)\) and \((j,i)\) entries and 0 elsewhere.
Whenever \(\b g^*_{ij}\in(0,\bar w)\), we can choose sufficiently small \(\eta>0\) so that \(\b g'=\b g^*\pm\eta\b e_{ij}\) are in \(\mathcal G_n\). Since \begin{align*}&\langle\left\{\phi[\b I-\phi\b g^*]^{-1}\b a^*\b a^{*T}[\b I-\phi\b g^*]^{-2}-\mu^*\kappa(\b g^*-\hb g) \right\},\eta\b e_{ij}\rangle\\&\qquad=-\langle \left\{\phi[\b I-\phi\b g^*]^{-1}\b a^*\b a^{*T}[\b I-\phi\b g^*]^{-2}-\mu^*\kappa(\b g^*-\hb g) \right\},-\eta\b e_{ij}\rangle,\end{align*}
we must have \[\langle \left\{\phi[\b I-\phi\b g^*]^{-1}\b a^*\b a^{*T}[\b I-\phi\b g^*]^{-2}-\mu^*\kappa(\b g^*-\hb g) \right\},\b e_{ij}\rangle=0.\]
Expanding the inner product gives the first case of (\ref{eq-foc2}) and similar arguments give the rest. 

\end{proof}

\begin{proof}{Proposition \ref{th-1}}
Parts (a) to (d) are derived in the main text, while part (e) is obtained by summing the result in (d) across all \(g_{ij}\).
\end{proof}

\strut

\begin{proof}{Proposition \ref{lem-ext2}}
    Lower bound: Suppose that \(\b a^*,\b g^*\) is an optimal solution to the problem \(\max V(\b a,\b g; \hb g,\hb a = 0,(\sqrt{C}-\|\hb a\|)^2)\). Then, by the triangle inequality,
    \begin{equation*}
        \begin{split}
            \|\b a^*-\hb a\|^2 & \leq (\|\b a^*\|+\|\hb a\|)^2 \\
            & \leq (\sqrt{(\sqrt{C}-\|\hb a\|)^2- \kappa\|\b g^*-\hb g\|^2}+\|\hb a\|)^2 \\
            & =(\sqrt{C}-\|\hb a|)^2+2\|\hb a\|\sqrt{(\sqrt{C}-\|\hb a\|)^2- \kappa\|\b g^*-\hb g\|^2}+\|\hb a\|^2-\kappa\|\b g^*-\hb g\|^2 \\
            & \leq (\sqrt{C}-\|\hb a\|)^2+2\|\hb a\|(\sqrt{C}-\|\hb a\|)+\|\hb a\|^2-\kappa\|\b g^*-\hb g\|^2 \\
            & = C - \kappa\|\b g^*-\hb g\|^2.
        \end{split}
    \end{equation*}
    so \(\b a^*,\b g^*\) is a feasible intervention for the problem \(\max V(\b a,\b g;\hb g,\hb a,C)\). Hence \(V^*(\hb g,\hb a,C)\geq V^*(\hb g,0,(\sqrt{C}-\|\hb a\|)^2)\).

    Upper bound: Suppose that \(\b a^{**},\b g^{**}\) is an optimal solution to the problem \(\max V(\b a,\b g; \hb g,\hb a ,C)\). Then, by the triangle inequality,
    \begin{equation*}
            \|\b a^{**}\|-\|\hb a\| \leq \|\b a^{**}-\hb a\|  \leq \sqrt{C-\kappa\|\b g^{**}-\hb g\|^2}.
    \end{equation*}
    Therefore, 
    \begin{equation*}
    \begin{split}
        \|\b a^{**}\|^2 & \leq (\sqrt{C-\kappa\|\b g^{**}-\hb g\|^2}+\|\hb a\|)^2 \\
        & =C+2\|\hb a\|\sqrt{C-\kappa\|\b g^{**}-\hb g\|^2}+\|\hb a\|^2 - \kappa\|\b g^{**}-\hb g\|^2        \\ 
        &\leq C+2\|\hb a\|\sqrt{C}+\|\hb a\|^2 - \kappa\|\b g^{**}-\hb g\|^2 \\
        & = (\sqrt{C}+\|\hb a\|)^2 - \kappa\|\b g^{**}-\hb g\|^2.
    \end{split}
    \end{equation*}
     so \(\b a^{**},\b g^{**}\) is a feasible intervention for the problem \(\max V(\b a,\b g;\hb g,\hb a=0,(\sqrt{C}+\|\hb a\|)^2)\). Hence \( V^*(\hb g,\hb a,C)\leq V^*(\hb g,0,(\sqrt{C}+\|\hb a\|)^2)\).

     Now we prove the second part of this proposition. By the first part of this proposition, for $C\geq \|\hb a\|^2$,
     \begin{equation}\label{eq-ext-1}
         \frac{V^*(\hb g,\hb a,C)}{V^*(\hb g,0,(\sqrt{C}-\|\hb a\|)^2)}\leq \frac{V^*(\hb g,0,(\sqrt{C}+\|\hb a\|)^2)}{V^*(\hb g,0,(\sqrt{C}-\|\hb a\|)^2)}=\frac{f((\sqrt{C}+\|\hb a\|)^2)}{f((\sqrt{C}-\|\hb a\|)^2)}.
     \end{equation}
     Also, by Proposition \ref{th-1}, we have that $f(x)$ is convex. Suppose $\b g^{***}$ is the optimal network solution to the problem $\max V(\b a,\b g,\hb g,0,(C+\|\hb a\|)^2)$. Thus, by the envelope theorem, 
     \begin{equation}\label{eq-ext-2}
         \begin{split}
             f((\sqrt{C}-\|\hb a\|)^2) & \geq ((\sqrt{C}-\|\hb a\|)^2-(\sqrt{C}+\|\hb a\|)^2)f'((\sqrt{C}+\|\hb a\|)^2)+f((\sqrt{C}+\|\hb a\|)^2) \\
             &=  ((\sqrt{C}-\|\hb a\|)^2-(\sqrt{C}+\|\hb a\|)^2)\frac{1}{(1-\lambda_1(\phi\b g^{***}))^2}+f((\sqrt{C}+\|\hb a\|)^2).
         \end{split}
     \end{equation}
     Therefore, by (\ref{eq-ext-1}), (\ref{eq-ext-2}), and $f((\sqrt{C}+\|\hb a\|)^2)=\frac{(\sqrt{C}+\|\hb a\|)^2-\kappa\|\b g^{***}-\hb g\|^2}{(1-\lambda_1(\phi\b g^{***}))^2}$,
     \begin{equation*}
         \begin{split}
              \frac{V^*(\hb g,\hb a,C)}{V^*(\hb g,0,(\sqrt{C}-\|\hb a\|)^2)} & \leq \frac{f((\sqrt{C}+\|\hb a\|)^2)}{f((\sqrt{C}-\|\hb a\|)^2)} \\
              & \leq \frac{f((\sqrt{C}+\|\hb a\|)^2)}{-\frac{4\sqrt{C}\|\hb a\|}{(1-\lambda_1(\phi\b g^{***}))^2}+f((\sqrt{C}+\|\hb a\|)^2)} \\
              & = \frac{(\sqrt{C}+\|\hb a\|)^2-\kappa\|\b g^{***}-\hb g\|^2}{-4\sqrt{C}\|\hb a\|+(\sqrt{C}+\|\hb a\|)^2-\kappa\|\b g^{***}-\hb g\|^2} \\
              & = 1+\frac{4\sqrt{C}\|\hb a\|}{(\sqrt{C}-\|\hb a\|)^2-\kappa\|\b g^{***}-\hb g\|^2}.
         \end{split}
     \end{equation*}
\end{proof}

\begin{lemma}\label{lem-graph}
Let \(\b g\in\mathcal G_n\).

(i) \[\lmax(\b g)\leq\bar w(n-1),\] with equality if and only if \(\b g\) is the complete graph \(\bar wK_n\).

(ii) \[\lmin(\b g)\geq-\bar w\sqrt{\left\lfloor\frac{n}{2}\right\rfloor\left\lceil\frac{n}{2}\right\rceil},\]
with equality if and only if \(\b g\) is isomorphic to the complete bipartite graph \(\bar wK_{\lfloor\frac{n}{2}\rfloor,\lceil\frac{n}{2}\rceil}\).
\end{lemma}

\begin{proof}{Lemma \ref{lem-graph}}

(i) Let \(\b g\in\mathcal G_n\). Let \(\b u=\b u_1(\b g)\).\footnote{Since \(\b g\) is nonnegative, such a nonnegative eigenvector exists by the Perron-Frobenius theorem.} Pick any \(u_k=\max_i u_i>0\). Then 
\[\lmax(\b g)u_k=(\b g\b u)_k= \sum_{i=1}^ng_{ki}u_i\leq \bar w(n-1)u_k,\]
with equality only if \(g_{ki}=\bar w\) and \(u_i=u_k\) for all \(i\neq k\). The latter implies that our choice of \(k\) can be replaced by any other \(j\), so we have \(g_{ji}=\bar w\) for all \(i\neq j\). Hence \(\b g\) represents \(\bar w K_n\).

(ii) We begin by stating Proposition 7 of \cite{bkd}:
\begin{proposition*}[\cite{bkd}]
Let \(\b g\) be a simple graph. Let \(\b u\) be an eigenvector for \(\lmin(\b g)\) and let \(R=\{i:u_i\geq0\},\ S=\{j:u_j<0\}\). Construct \(\b g'\) by removing links within \(R\) and \(S\), and adding links between \(R\) and \(S\). Then \(\lmin(\b g')\leq\lmin(\b g)\).
\end{proposition*}

\textbf{Proof.} We have \begin{align}\lmin(\b g)&=\sum_{i,j\in R}u_iu_jg_{ij}+\sum_{i,j\in S}u_iu_jg_{ij}+2\sum_{i\in R,j\in S}u_iu_jg_{ij}\nonumber\\&\geq\sum_{i,j\in R}u_iu_jg'_{ij}+\sum_{i,j\in S}u_iu_jg'_{ij}+2\sum_{i\in R,j\in S}u_iu_jg'_{ij}=\lmin(\b g'),\label{eq-bkd}\end{align}

Clearly, the same argument applies even if \(\b g\) is allowed to be a weighted graph, so a complete bipartite graph is optimal. Furthermore, among the set of complete bipartite graphs, the smallest eigenvalue occurs when the vertices are partitioned into sets of size \(\lfloor\frac{n}{2}\rfloor\) and \(\lceil\frac{n}{2}\rceil\). It remains to show that \(\bar wK_{\lfloor\frac{n}{2}\rfloor,\lceil\frac{n}{2}\rceil}\) is the unique graph (up to isomorphism) that minimizes \(\lmin(\b g)\), with \[\lmin(\bar wK_{\lfloor\frac{n}{2}\rfloor,\lceil\frac{n}{2}\rceil})=-\bar w\sqrt{\left\lfloor\frac{n}{2}\right\rfloor\left\lceil\frac{n}{2}\right\rceil}.\]

Let \(\b g\) be a network that is not isomorphic to \(\bar wK_{\lfloor\frac{n}{2}\rfloor,\lceil\frac{n}{2}\rceil}\).
First suppose that \(u_i\neq 0\) for all \(i\). Then the inequality in \eqref{eq-bkd} holds strictly, so there exists \(\b g'\) with \(\lmin(\b g')<\lmin(\b g)\), thus \(\b g\) cannot be optimal. 

Otherwise, without loss of generality suppose that \(u_n=0\). Let \(\b g_{n-1}\) be the \((n-1)\)-th principal minor of \(\b g\), and \(\b u_{1:n-1}\) be the first \(n-1\) components of \(\b u\). Then \[\b g\b u=\lmin(\b g)\b u\implies \b g_{n-1}\b u_{1:n-1}=\lmin(\b g)\b u_{1:n-1},\] so \(\lmin(\b g)\) is also an eigenvalue of \(\b g_{n-1}\). This implies that \[\lmin(\b g)\geq\lambda_{n-1}(\b g_{n-1})\geq -\bar w\sqrt{\left\lfloor\frac{n-1}{2}\right\rfloor\left\lceil\frac{n-1}{2}\right\rceil}>-\bar w\sqrt{\left\lfloor\frac{n}{2}\right\rfloor\left\lceil\frac{n}{2}\right\rceil},\]

so \(\b g\) also cannot be optimal. Hence the only minimizers of \(\lmin(\b g)\) are isomorphic to \(\bar wK_{\lfloor\frac{n}{2}\rfloor,\lceil\frac{n}{2}\rceil}\).
\end{proof}

\strut

\begin{proof}{Remark \ref{rmk-2}}From the bounds in Lemma \ref{lem-graph},

(a) If \(\phi>0\), then \(\lmax(\phi\b g)\leq\phi\bar w(n-1)<1.\)

(b) If \(\phi<0\) and \(2\mid n\), then \[\lmax(\phi\b g) = \phi\lmin(\phi\b g) \leq-\phi\bar w\sqrt{\left\lfloor\frac{n}{2}\right\rfloor\left\lceil\frac{n}{2}\right\rceil}=-\phi\bar w\frac{n}{2}<1.\]

(c) If \(\phi<0\) and \(2\nmid n\), then \[\lmax(\phi\b g) = \phi\lmin(\phi\b g) \leq-\phi\bar w\sqrt{\left\lfloor\frac{n}{2}\right\rfloor\left\lceil\frac{n}{2}\right\rceil}=-\phi\bar w\sqrt{\frac{n^2-1}{4}}<1.\]
\end{proof}

\begin{proof}{Theorem \ref{th-2}}
To prove the existence of a cutoff \(\bar C\), we first take limits of \eqref{eq-foca}: 
\begin{align*}&\lim_{C\to\infty}2[\b I-\phi\b g^*]^{-2}\frac{\b a^*}{\sqrt C}=\lim_{C\to\infty}2\mu\frac{\b a^*-\hb a}{\sqrt C}\\\implies&\lim_{C\to\infty}[\b I-\phi\overline{\b g}]^{-2}\frac{\b a^*}{\sqrt C}=\lim_{C\to\infty}\mu\frac{\b a^*}{\sqrt C},\end{align*} so \[\lim_{C\to\infty}\mu=\frac{1}{(1-\overline\lambda)^2}.\] Similar to \cite{ggg}, \(\frac{\b a^*}{\sqrt C}\) goes to the corresponding eigenvector \(\b u(\overline{\b g})\). From (\ref{eq-foc2}), if there exists arbitrary large \(C\) such that \(g_{kl}^*\in(0,\bar w)\), we have
\begin{align*}0=&\lim_{C\to\infty}\frac{2\kappa(\b g^*-\hb g)_{kl}}{C}\\=&\lim_{C\to\infty}\frac{1}{\mu C}(\phi[\b I-\phi\b g^*]^{-1}\b a^*\b a^{*T}[\b I-\phi\b g^*]^{-2}+\phi[\b I-\phi\b g^*]^{-2}\b a^*\b a^{*T}[\b I-\phi\b g^*]^{-1})_{kl}\\=&\ \phi(1-\overline\lambda)^2\big([\b I-\phi\overline{\b g}]^{-1}\b u(\overline{\b g})\b u(\overline{\b g})^T[\b I-\phi\overline{\b g}]^{-2}+[\b I-\phi\overline{\b g}]^{-2}\b u(\overline{\b g})\b u(\overline{\b g})^T[\b I-\phi\overline{\b g}]^{-1}\big)_{kl}
\\=&\ \frac{2\phi}{1-\overline\lambda}u_k(\overline{\b g})u_l(\overline{\b g})\\\neq&\ 0,\end{align*}
with the last inequality because \(u_k(\overline{\b g})\neq0\) for all \(k\). Therefore, there cannot be interior \(g^*_{ij}\) for sufficiently large \(C\), so \(\b g^*\) must be either complete or complete bipartite from Lemma \ref{lem-graph}.
\end{proof}

\begin{proof}{Fact \ref{fact-1}}
(a) It is easy to check that \((1,1,\cdots,1)\) is an eigenvector of \(K_p\). By the Perron-Frobenius theorem, it must also be a basis of the eigenspace of \(\lmax(K_p)\).

(b) We note that \[K_{p,q}=\pmat{\b 0_p&\b J_{pq}\\\b J_{qp}&\b 0_q}\] is of rank two and has zero trace, so it has a unique eigenvector that corresponds to a negative eigenvalue. We can verify that the given vector is the desired eigenvector of \(\lambda_{p+q}(K_{p,q})\).
\end{proof}

\begin{proof}{Proposition \ref{prop-np}} Call the constrained version of MAX-CUT with \(|S|=\lfloor\frac{n}{2}\rfloor\) the balanced maximum cut (BAL-MAX-CUT) problem, and call a partition of \(\mathcal N\) into parts of sizes \(\lceil\frac{n}{2}\rceil\) and \(\lceil\frac{n}{2}\rceil\) a balanced cut.

MAX-CUT \(\leq_P\) BAL-MAX-CUT:\footnote{We write \(X\leq_P Y\) if problem \(X\) is reducible to problem \(Y\) in polynomial time.} Given an instance \(G\) of MAX-CUT with adjacency matrix \(\b m_{p\times p}\), consider the instance \(G'\) of BAL-MAX-CUT with adjacency matrix \(\pmat{\b m&\b 0_p\\\b 0_p&\b 0_p}\).

Then every cut of \(G\) can be extended to a balanced cut of \(G'\) by a suitable assignment of the independent vertices, without changing the total cut weight. Similarly, every balanced cut of \(G'\) can be restricted to a cut of \(G\) without changing the cut weight by removing the additional vertices. Thus the instance \(G'\) of BAL-MAX-CUT solves the MAX-CUT problem.

BAL-MAX-CUT \(\leq_P\) MAX-CUT: Given an instance \(H\) of BAL-MAX-CUT with adjacency matrix \(\b m_{p\times p}\), consider an instance \(H'\) of BAL-MAX-CUT with adjacency matrix \(\b m+\alpha(\b J_{pp}-\b I_p)\), where \(\alpha>\b 1^T_p\b m\b 1_p\) is sufficiently large.

Let \(k=\left\lfloor\frac{p}{2}\right\rfloor\left\lceil\frac{p}{2}\right\rceil\) be the number of edges in a half-cut of \(H'\). Then the weight of any  balanced cut is at least \(\alpha k\), while any other cut has at most \(k-1\) edges so has weight at most \(\alpha(k-1)+\b 1^T_p\b m\b 1_p<\alpha k\). Therefore, the maximal cut is the maximal balanced cut and the instance \(H'\) of MAX-CUT solves the BAL-MAX-CUT problem.

Therefore, BAL-MAX-CUT, and hence the orientation problem, is in the same computational class as the MAX-CUT problem and is NP-hard \citep{karp}.
\end{proof}

\begin{proof}{Theorem \ref{th-welfare}}
We have \[\lim_{C\to\infty} r^*(\hb g,\hb a,C)=\lim_{C\to\infty}\max_{\b g\in\mathcal G_n}\frac{V^*_{single}(\b g,\hb a,C)}{V^*_{single}(\hb g,\hb a,C)}=\max_{\b g\in\mathcal G_n}\left(\frac{1-\lambda_1(\phi\hb g)}{1-\lambda_1(\phi\b g)}\right)^2,\]
with \(\max_{\b g\in\mathcal G_n}\lambda_1(\phi\b g)\) given by Lemma \ref{lem-graph}.
\end{proof}

\begin{proof}{Theorem \ref{th-ineq}}
Part (b) is shown in  Example \ref{ex-4}, while part (a) follows directly from Lemma \ref{lem-2} and discussions in the main text.
\end{proof}

\begin{proof}{Proposition \ref{lem-2}}
The proof largely follows from the main text. It remains to justify that \[
\frac{x_i^{*}}{x_j^{*}} \approx  \frac{a_i^*}{a_j^*} \approx \frac{u_i^1}{u_j^1} \text{ for all }i,j\]
when \(C\) is large. By \eqref{eq-foc1}, and possibly multiplying \(\b u^1(\phi\b g^*)\) by \(-1\), we have the relation \[\lim_{C\to\infty}\frac{\b a^*}{\|\b a^*\|}=\lim_{C\to\infty}\b u^1(\phi\b g^*).\]
Therefore, by \eqref{eq-x*} and the above,
\[\lim_{C\to\infty}\frac{\b x^*}{\|\b a^*\|}=\lim_{C\to\infty}\frac{[\b I-\phi\b g]^{-1}\b a^*}{\|\b a^*\|}=\lim_{C\to\infty}[\b I-\phi\b g]^{-1}\b u^1(\phi\b g^*)=\lim_{C\to\infty}\frac{1}{1-\lambda_1(\phi\b g)}\b u^1(\phi\b g^*).\]
Consequently, \(\b x^*, \b a^*,\b u^1(\phi\b g^*)\) are approximately proportional vectors when \(C\) is large and the desired equation holds.
\end{proof}

\begin{proof}{Lemma \ref{lem-3}}
Since \(|u^n_i|=|u^n_j|\) for all \(i,j\), then \(|u^n_i|=\frac{1}{\sqrt n}\) for all \(i\). By a relabelling of the indices and possibly multiplying by \(-1\), without loss of generality let \(u^n_i=\frac{1}{\sqrt n}\) if \(i\in\{1,\cdots,k\}\), and \(u^n_i=-\frac{1}{\sqrt n}\) otherwise. Also let \(k>\frac{n}{2}\). By definition, \[\lmin(\b g)u^n_1=\sum_{i=1}^ng_{1i}u^n_i=\sum_{i=1}^kg_{1i}u^n_1-\sum_{i=k+1}^ng_{1i}u^n_1\geq -\bar w(n-k)u^n_1\geq -\bar w\left(\frac{n-1}{2}\right)u^n_1,\]
so \(\lmin(\b g)\geq -\bar w\left(\frac{n-1}{2}\right)\). Finally, it is easily verified that equality holds under the given choice of \(\b g\).
\end{proof}

\begin{proof}{Proposition \ref{prop-4}}
For zero inequality, we must have \(k\b z=\b x^*=[\b I-\phi\hb g]^{-1}\b a^*\) for some \(k\in\mathbb R\). Thus \(\b a^*=k[\b I-\phi\hb g]\b z\). By the budget constraint, \(\|\b a^*\|^2=C=k^2\|[\b I-\phi\hb g]\b z\|^2\), so  \[V_{single,eq}^*=(\b a^*)^T[\b I-\phi\hb g]^{-2}\b a=k^2=\frac{C}{\|[\b I-\phi\hb g]\b z\|^2}.\]
\end{proof}

\newpage
\bibliographystyle{chicago}
\bibliography{GGG-g}
\end{document}